\documentclass[11pt]{article}
\usepackage[margin=1in]{geometry}
\usepackage{etoolbox}
\usepackage[fleqn]{amsmath}
\usepackage{amssymb,bbm,amsthm}
\usepackage{thm-restate,thmtools}
\usepackage{mathtools}
\usepackage{booktabs}
\usepackage{makecell}
\usepackage{subcaption}
\usepackage{tabularx}

\usepackage{enumitem}
\setlist{
  listparindent=\parindent,
}

\usepackage[shortcuts]{extdash}

\usepackage{xcolor}
\usepackage{nameref}
\definecolor[named]{Purple}{cmyk}{0.56,1,0,0.15}
\definecolor[named]{DarkBlue}{cmyk}{1,0.6,0,0.2}
\usepackage[
  linktocpage=true,
  pagebackref=true,
  colorlinks=false,
  linkcolor=Purple,
  citecolor=Purple,
  urlcolor=DarkBlue,
  bookmarks,
  bookmarksopen,
  bookmarksnumbered
]{hyperref}
\usepackage[capitalize,noabbrev]{cleveref}

\declaretheorem[numberwithin=section]{theorem}
\declaretheorem[unnumbered, name=Theorem]{theorem*}
\declaretheorem[numberlike=theorem]{lemma}

\declaretheorem[numberlike=theorem]{corollary}

\declaretheorem[numberlike=theorem]{claim}

\declaretheorem[numberlike=theorem, name=Definition]{definition}
\declaretheorem[unnumbered, name=Definition]{definition*}

\declaretheorem[unnumbered, name=Conjecture]{conjecture*}

\declaretheorem[unnumbered, name=Hypothesis]{hypothesis*}
\declaretheorem[numberlike=theorem,name=Open Question]{question}

\usepackage{algorithm}
\usepackage{algorithmicx}
\usepackage[noend]{algpseudocode}

\crefname{claim}{Claim}{Claims}
\Crefname{claim}{Claim}{Claims}

\DeclarePairedDelimiter\parens{\lparen}{\rparen}

\DeclarePairedDelimiter\set{\lbrace}{\rbrace}

\DeclarePairedDelimiter\ceil{\lceil}{\rceil}

\DeclarePairedDelimiter\abs{\lvert}{\rvert}

\DeclarePairedDelimiterX\innerprod[2]{\langle}{\rangle}{#1,#2}

\DeclarePairedDelimiterX\intervaloo[2]{(}{)}{#1\,.\,.\,#2}
\DeclarePairedDelimiterX\intervaloc[2]{(}{]}{#1\,.\,.\,#2}
\DeclarePairedDelimiterX\intervalco[2]{[}{)}{#1\,.\,.\,#2}
\DeclarePairedDelimiterX\intervalcc[2]{[}{]}{#1\,.\,.\,#2}

\undef\Pr
\DeclareMathOperator*\Pr{\mathbf P}

\DeclareMathOperator\One{\mathbf 1}

\DeclareMathOperator\HD{HD}

\DeclareMathOperator\MM{MM}

\DeclareMathOperator\poly{poly}

\newcommand\Order{O}
\newcommand\order{o}

\undef\mod
\undef\bmod
\DeclareMathOperator\bmod{mod}

\newcommand\mod[1]{\;\,({\bmod}\:#1)}

\def\prob#1{\relax\ifmmode\text{\textsc{#1}}\else\textsc{#1}\fi}
\def\alg#1{\relax\ifmmode\text{\textsc{#1}}\else\textsc{#1}\fi}

\usepackage{nicefrac}
\usepackage{setspace}
\renewcommand{\tilde}{\widetilde}

\newcommand\cs[1]{}
\newcommand\nick[1]{}
\newcommand\ron[1]{}
\newcommand\amir[1]{}
\newcommand\elazar[1]{}
\newcommand\todo[1]{}

\title{\textbf{Can You Solve Closest String Faster than Exhaustive~Search?%
}}
\author{Amir Abboud\footnote{This project has received funding from the European Research Council (ERC) under the European Union’s Horizon Europe research and innovation programme (grant agreement No 101078482). Additionally, Amir Abboud is supported by an Alon scholarship and a research grant from the Center for New Scientists at the Weizmann Institute of Science.
}\\ Weizmann Institute of Science\\\texttt{amir.abboud@weizmann.ac.il} \and  Nick Fischer\\ Weizmann Institute of Science\\\texttt{nick.fischer@weizmann.ac.il} \and Elazar Goldenberg\\ Academic College of Tel Aviv-Yaffo\\\texttt{elazargo@mta.ac.il}\and  Karthik C.\ S.\footnote{This work was supported by a grant from the Simons Foundation, Grant Number 825876, Awardee Thu D.
Nguyen.}\\ Rutgers University\\\texttt{karthik0112358@gmail.com} \and Ron Safier\\ Weizmann Institute of Science\\\texttt{ron.safier@weizmann.ac.il	}}
\date{}

\begin{document}
\maketitle
\begin{abstract}
\noindent
We study the fundamental problem of finding the best string to represent a given set, in the form of the \emph{Closest String} problem: Given a set $X \subseteq \Sigma^d$ of $n$ strings, find the string $x^*$ minimizing the radius of the smallest Hamming ball around $x^*$ that encloses all the strings in $X$. In this paper, we investigate whether the Closest String problem admits algorithms that are faster than the trivial exhaustive search algorithm. We obtain the following results for the two natural versions of the problem:
\begin{itemize}
    \item In the \emph{continuous} Closest String problem, the goal is to find the solution string $x^*$ anywhere in $\Sigma^d$. For binary strings, the exhaustive search algorithm runs in time $\Order(2^d \poly(nd))$ and we prove that it cannot be improved to time $\Order(2^{(1-\epsilon) d} \poly(nd))$, for any $\epsilon > 0$, unless the Strong Exponential Time Hypothesis fails. 
    \item In the \emph{discrete} Closest String problem, $x^*$ is required to be  in the input set $X$. While this problem is clearly in polynomial time, its fine-grained complexity has been pinpointed to be quadratic time~\smash{$n^{2\pm \order(1)}$} whenever the dimension is $\omega(\log n) < d < n^{\order(1)}$. We complement this known hardness result with new algorithms, proving essentially that whenever $d$ falls out of this hard range, the discrete Closest String problem can be solved faster than exhaustive search. In the small-$d$ regime, our algorithm is based on a novel application of the inclusion-exclusion principle. 
\end{itemize}
Interestingly, all of our results apply (and some are even stronger) to the natural dual of the Closest String problem, called the \emph{Remotest String} problem, where the task is to find a string maximizing the Hamming distance to all the strings in $X$.
\end{abstract}

\clearpage

\section{Introduction} \label{sec:introduction}

The challenge of characterizing a set of strings by a single \emph{representative} string is a fundamental problem all across computer science, arising in essentially all contexts where strings are involved.
The basic task is to find a string~$x^*$ which minimizes the maximum number of mismatches to all strings in a given set~$X$. 
Equivalently, the goal is to find the center $x^*$ of a smallest ball enclosing all strings in $X$ in the Hamming (or~$\ell_0$) metric. This problem has been studied under various names, including \emph{Closest String}, \emph{1-Center in the Hamming metric} and \emph{Chebyshev Radius}, and constitutes the perhaps most elementary clustering task for strings. \nick{@Karthik: You wanted to add a sentence like this? More broadly, identifying center points in metric spaces has puzzled mathematicians for centuries.}

In the literature, the Closest String problem has received a lot of attention~\cite{FrancesL97,GasieniecJL99,LiMW02,GrammNR03,LanctotLMWZ03,MauchMH03,MenesesLOP04,MaS09,KochmanMP12,MazumdarPS13,AbboudBCCS21}, and it is not surprising that besides the strong theoretical interest, it finds wide-reaching applications in various domains including machine learning, bioinformatics, coding theory and cryptography. One such application in machine learning is for clustering \emph{categorical} data. Typical clustering objectives involve finding good center points to characterize a set of feature vectors. For numerical data (such as a number of publications) this task translates to a center (or median) problem over, say, the $\ell_1$ metric which can be solved using geometry tools. For categorical data, on the other hand, the points have non-numerical features (such as blood type or nationality) and the task becomes finding a good center string over the Hamming metric.

Another important application, in the context of computational biology, is the computer-aided design of PCR primers~\cite{MaS09,LucasBMT91,DopazoRSS93,ProutskiH96,GrammHN02,WangCLC06}. On a high level, in the PCR method the goal is to find and amplify (i.e., copy millions of times) a certain fragment of some sample DNA. To this end, short DNA fragments (typically~18 to 25 nucleotides) called \emph{primers} are used to identify the start and end of the region to be copied. These fragments should match as closely as possible the target regions in the sample DNA. Designing such primers is a computational task that reduces exactly to finding a closest string in a given set of genomes.

The Closest String problem comes in two different flavors: In the \emph{continuous} Closest String problem the goal is to select an \emph{arbitrary} center string $x^* \in \Sigma^d$ (here, $\Sigma$ is the underlying alphabet) that minimizes the maximum Hamming distance to the $n$ strings in $X$. This leads to a baseline algorithm running in exponential time~\makebox{$\Order(|\Sigma|^d \poly(nd))$}. In the \emph{discrete} Closest String problem, in contrast, the task is to select the best center~$x^*$ in the given set of strings~$X$; this problem therefore admits a baseline algorithm in time~$\Order(n^2 d)$. Despite the remarkable attention that both variants have received so far, the most basic questions about the continuous and discrete Closest String problems have not been fully resolved yet:

\smallskip
\begin{center}
    \emph{Can the $\Order(|\Sigma|^d \poly(nd))$-time algorithm for continuous Closest String be improved?}\\[.2ex]
    \emph{Can the $\Order(n^2 d)$-time algorithm for discrete Closest String be improved?}
\end{center}
\medskip
In this paper, we make considerable progress towards resolving both driving questions, by respectively providing tight conditional lower bounds and new algorithms. In the upcoming \cref{sec:introduction:sec:continuous,sec:introduction:sec:discrete} we will address these questions in depth and state our results.

\medskip
Interestingly, in both cases our results also extend, at times even in a stronger sense, to a natural dual of the Closest String problem called the \emph{Remotest String} problem. Here, the task is to find a string $x^*$ that maximizes the minimum Hamming distance from $x^*$ to a given set of strings~$X$. This problem has also been studied in computational biology~\cite{LanctotLMWZ03,Lanctot04} and more prominently in the context of coding theory: The remotest string distance is a fundamental parameter of any code which is also called the \emph{covering radius}~\cite{CohenHLL97}, and under this name the Remotest String problem has been studied in previous works~\cite{Micciancio04,GuruswamiMR04,AlonPY09,HavivR12} mostly for specific sets $X$ such as linear codes or lattices. See Alon, Panigrahy and Yekhanin~\cite{AlonPY09} for further connections to matrix rigidity.

\subsection{Continuous Closest/Remotest String} \label{sec:introduction:sec:continuous}
Let us start with the more classical \emph{continuous} Closest String problem. It is well-known that the problem is NP-complete~\cite{FrancesL97,LanctotLMWZ03}, and up to date the best algorithm remains the naive one: Exhaustively search through all possible strings in time $\Order(|\Sigma|^d \poly(nd))$. This has motivated the study of approximation algorithms leading to various approximation schemes~\cite{GasieniecJL99,LiMW02,MaS09,MazumdarPS13}, and also the study through the lens of parameterized algorithms~\cite{GrammNR03}. In this work, we insist on exact algorithms and raise again the question: \emph{Can you solve the continuous Closest String problem faster than exhaustive search?}

For starters, focus on the Closest String problem for \emph{binary} alphabets (i.e., for $|\Sigma| = 2$) which is of particular importance in the context of coding theory~\cite{KochmanMP12}. From the known NP-hardness reduction which is based on the 3-SAT problem~\cite{FrancesL97}, it is not hard to derive a $2^{d/2}$ lower bound under the Strong Exponential Time Hypothesis (SETH)~\cite{ImpagliazzoP01,ImpagliazzoPZ01}. This bound clearly does not match the upper bound and possibly leaves hope for a meet-in-the-middle-type algorithm. In our first contribution we shatter all such hopes by strengthening the lower bound,
with considerably more effort, to match the time complexity of exhaustive search:

\begin{restatable}[Continuous Closest String is SETH-Hard]{theorem}{corcontinuousclosestseth} \label{cor:continuous-closest-seth}
The continuous Closest String problem on binary alphabet cannot be solved in time $\Order(2^{(1-\epsilon) d} \poly(n))$, for any $\epsilon > 0$, unless SETH fails.
\end{restatable}

Interestingly, we obtain this lower bound as a corollary of the analogous lower bound for the continuous Remotest String problem (see the following \cref{thm:continuous-remotest-seth}). This is because both problems are equivalent over the binary alphabet. However, even for larger sized alphabet sets $\Sigma$, we obtain a matching lower bound against the Remotest String problem: 

\begin{restatable}[Continuous Remotest String is SETH-Hard]{theorem}{thmcontinuousremotestseth} \label{thm:continuous-remotest-seth}
The continuous Remotest String problem on alphabet set $\Sigma$ cannot be solved in time $\Order(|\Sigma|^{(1-\epsilon) d} \poly(n))$, for any $\epsilon > 0$ and $|\Sigma| = \order(d)$, unless SETH fails.
\end{restatable}

\cref{thm:continuous-remotest-seth} gives a \emph{tight} lower bound for the continuous Remotest String problem in all regimes where we can expect lower bounds, and we therefore close the exact study of the continuous Remotest String problem. Indeed, in the regime where the alphabet size $|\Sigma|$ exceeds the dimension $d$, the Closest and Remotest String problems \emph{can} be solved faster in time $\Order(d^d \poly(n, d))$ (and even faster parameterized in terms of the target distance~\cite{GrammNR03}). 

The intuition behind \cref{thm:continuous-remotest-seth} is simple: We encode a $k$-SAT instance as a Remotest String problem by viewing strings as \emph{assignments} and by searching for a string which is remote from all \emph{falsifying} assignments. The previously known encoding~\cite{FrancesL97} was inefficient (encoding a single variable $X_i$ accounted for \emph{two} letters in the Remotest String instance: one for encoding the truth value and another one as a ``don't care'' value for clauses not containing $X_i$), and our contribution is that we make the encoding lossless. While superficially simple, this baseline idea requires a lot of technical effort.

\subsection{Discrete Closest/Remotest String} \label{sec:introduction:sec:discrete}
Recall that in the \emph{discrete} Closest String problem (in contrast to the continuous one) the solution string~$x^*$ must be part of the input set $X$. For applications in the context of data compression and summarization, the discrete problem is often the better choice: Selecting the representative string from a set of, say, grammatically or semantically meaningful strings is typically more informative than selecting an arbitrary representative string.

The problem can be naively solved in time $\Order(n^2 d)$ by exhaustive search: Compute the Hamming distance between all $\binom{n}{2}$ pairs of strings in $X$ in time $\Order(d)$ each. In terms of exact algorithms, this running time is the fastest known. Toward our second driving question, we investigate whether this algorithm can be improved, at least for some settings of $n$ and $d$. In previous work, Abboud, Bateni, Cohen-Addad, Karthik, and Seddighin~\cite{AbboudBCCS21} have established a conditional lower bound under the Hitting Set Conjecture~\cite{AbboudWW16}, stating that the problem requires quadratic time in $n$ whenever $d = \omega(\log n)$:

\begin{theorem}[Discrete Closest String for Super-Logarithmic Dimensions~\cite{AbboudBCCS21}] \label{thm:discrete-closest-hsc}
The discrete Closest String problem in dimension $d = \omega(\log n)$ cannot be solved in time $\Order(n^{2-\epsilon})$, for any $\epsilon > 0$, unless the Hitting Set Conjecture fails.
\end{theorem}

This hardness result implies that there is likely no polynomially faster algorithm for Closest String whenever the dimension~$d$ falls in the range $\omega(\log n) < d < n^{\order(1)}$. But this leaves open the important question of whether the exhaustive search algorithm can be improved outside this region, if $d$ is very small (say, $\order(\log n)$) or very large (i.e., polynomial in $n$). In this paper, we provide answers for both regimes.

\paragraph{Small Dimension}
Let us start with the small-dimension regime, $d = \order(\log n)$. The outcome of the question whether better algorithms are possible is a priori not clear. Many related center problems (for which the goal is to select a center point $x^*$ that is closest not necessarily in the Hamming metric but in some other metric space) differ substantially in this regard: On the one hand, in the Euclidian metric, even for $d = 2^{\Order(\log^* n)}$, the center problem requires quadratic time under the Hitting Set Conjecture~\cite{AbboudBCCS21}.\footnote{Technically, the problem is only known to be hard in the \emph{listing} version where we require to list \emph{all} feasible centers~\cite{AbboudBCCS21}.} On the other hand, in stark contrast, the center problem for the $\ell_1$ and~$\ell_\infty$ metrics can be solved in almost-linear time $n^{1+\order(1)}$ whenever the dimension is~$d = \order(\log n)$. This dichotomy phenomenon extends to even more general problems including nearest and furthest neighbor questions for various metrics and the maximum inner product problem~\cite{Williams18,Chen20}.

In view of this, we obtain the perhaps surprising result that whenever $d = \order(\log n)$ the discrete Closest String problem can indeed be solved in subquadratic---even almost-linear---time. More generally, we obtain the following algorithm:

\begin{restatable}[Discrete Closest String for Small Dimensions]{theorem}{thminclexclclosest} \label{thm:incl-excl-closest}
The discrete Closest String problem can be solved in time $\Order(n \cdot 2^d)$.
\end{restatable}

Note that this result is trivial for binary alphabets, and our contribution lies in finding an algorithm in time $\Order(n \cdot 2^d)$ for alphabets of \emph{arbitrary} size.

We believe that this result is interesting also from a technical perspective, as it crucially relies on the \emph{inclusion-exclusion} principle. While this technique is part of the everyday toolset for exponential-time and parameterized algorithms, it is uncommon to find applications for polynomial-time problems and our algorithm yields the first such application to a center-type problem, to the best of our knowledge. We believe that our characterization of the Hamming distance in terms of an inclusion-exclusion-type formula (see \cref{lem:hd-inclusion-exclusion}) is very natural and likely to find applications in different contexts.

\paragraph{Large Dimension}
In the large-dimension regime, where $d$ is polynomial in $n$, it is folklore that fast matrix multiplication should be of use. Specifically, over a binary alphabet we can solve the Closest String problem in time $\Order(\MM(n, d, n))$ (where $\MM(n, d, n)$ is the time to multiply an $n \times d$ by a $d \times n$ matrix) by using fast matrix multiplication to compute the Hamming distances between all pairs of vectors, rather than by brute-force. For arbitrary alphabet sizes this idea leads to a running time of $\Order(\MM(n, d|\Sigma|, n))$ which is of little use as $|\Sigma|$ can be as large as $n$ and in this case the running time becomes $\Omega(n^2 d)$.

We prove that nevertheless, the $\Order(n^2 d)$-time baseline algorithm can be improved using fast matrix multiplication---in fact, using ideas from \emph{sparse} matrix multiplication such as Yuster and Zwick's heavy-light idea~\cite{YusterZ05}.

\begin{restatable}[Discrete Closest String for Large Dimensions]{theorem}{thmfastmmclosest} \label{thm:fast-mm-closest}
For all $\delta > 0$, there is some~$\epsilon > 0$ such that the discrete Closest String problem with dimension $d = n^\delta$ can be solved in time $\Order(n^{2+\delta-\epsilon})$.
\end{restatable}

\paragraph{Remotest String}
Finally, we turn our attention to the discrete Remotest String problem. In light of the previously outlined equivalence in the continuous setting, we would expect that also in the discrete setting, the Closest and Remotest String problem are tightly connected. We confirm this suspicion and establish a strong equivalence for binary alphabets:

\begin{restatable}[Equivalence of Discrete Closest and Remotest String]{theorem}{thmdiscreteequiv} \label{thm:discrete-equiv}
If the discrete Closest String over a binary alphabet is in time $T(n, d)$, then the discrete Remotest String over a binary alphabet is in time $T(\Order(n), \Order(d + \log n)) + \widetilde\Order(nd)$. Conversely, if the discrete Remotest String over a binary alphabet is in time $T'(n, d)$, then the discrete Closest String over a binary alphabet is in time $T'(\Order(n), \Order(d + \log n)) + \widetilde\Order(nd)$.
\end{restatable}

In combination with \cref{thm:discrete-closest-hsc}, this equivalence entails that also Remotest String requires quadratic time in the regime $\omega(\log n) < d < n^{\order(1)}$. Let us remark that, while the analogous equivalence is trivial in the continuous regime, proving \cref{thm:discrete-equiv} is not trivial and involves the construction of a suitable gadget that capitalizes on explicit constant-weight codes.

The similarity between discrete Closest and Remotest String continues also on the positive side: All of our algorithms extend naturally to Remotest String, not only for binary alphabets (see \cref{thm:incl-excl-remotest,thm:fast-mm-remotest}).

\subsection{Open Problems}
Our work inspires some interesting open problems. The most pressing question from our perspective is whether there also is a $|\Sigma|^{(1-\order(1))d}$ lower-bound for continuous Remotest String (for alphabets of size bigger than $2$).

\begin{question}[Continuous Closest String for Large Alphabets]
For $|\Sigma| > 2$, can the continuous Closest String problem be solved in time $\Order(|\Sigma|^{(1-\epsilon)d} \poly(n))$, for some $\epsilon > 0$? 
\end{question}

We believe that our approach (proving hardness under SETH) hits a natural barrier for the Closest String problem. In some sense, the $k$-SAT problem behaves very similarly to Remotest String (with the goal to be remote from all falsifying assignments), and over binary alphabets remoteness and closeness can be exchanged. For larger alphabets this trivial equivalence simply does not hold. It would be exciting if this insight could fuel a faster \emph{algorithm} for Closest String, and we leave this question for future work.

\medskip
On the other hand, consider again the discrete Closest and Remotest String problems. While we close \emph{almost all} regimes of parameters, there is one regime which we did not address in this paper:

\begin{question}[Discrete Closest/Remotest String for Logarithmic Dimension]
Let $c$ be a constant. Can the discrete Closest and Remotest String problems with dimension $d = c \log n$ be solved in time $\Order(n^{2-\epsilon})$, for some $\epsilon = \epsilon(c) > 0$?
\end{question}

In the regime $d = \Theta(\log n)$, we typically expect only very sophisticated algorithms, say using the polynomial method in algorithm design~\cite{AbboudWY15}, to beat exhaustive search. And indeed, using the polynomial method it is possible to solve also discrete Closest and Remotest String in subquadratic time for binary (or more generally, constant-size) alphabets~\cite[Theorem~1.4]{AlmanW15,AlmanCW16}. The question remains whether subquadratic time complexity is also possible for unrestricted alphabet sizes.

\subsection{Outline of the Paper}
We organize this paper as follows. In \cref{sec:preliminaries} we give some preliminaries and state the formal definitions of the continuous/discrete Closest/Remotest String problems. In \cref{sec:continuous}   we prove our conditional hardness results for the continuous problems. In \cref{sec:discrete}   we treat in detail the discrete problems by designing our new algorithms and proving our fine-grained equivalence between Closest and Remotest String.

\section{Preliminaries} \label{sec:preliminaries}
We set $[n] = \set{1, \dots, n}$ and write $\widetilde\Order(T) = T (\log T)^{\Order(1)}$. We occasionally write $\One(P) \in \set{0, 1}$ to express the truth value of the proposition $P$.

\paragraph{Strings}
Let $\Sigma$ be a finite alphabet set of size at least $2$. For a string $x \in \Sigma^d$ of length (or dimension) $d$, we write $x[i]$ for the $i$-th character in $x$. For a subset $I \subseteq [d]$, we write $x[I] \in \Sigma^I$ for the subsequence obtained from $x$ by restricting to the characters in $I$. The \emph{Hamming distance} between two equal-length strings $x, y \in \Sigma^d$ is defined as $\HD(x, y) = \abs{\set{i \in [d] : x[i] \neq y[i]}}$. Let $X$ be a set of length-$d$ strings and let~$x^*$ be a length-$d$ string. Then we set
\begin{align*}
    r(x^*, X) &= \max_{y \in X} \HD(x^*, y) \quad\text{(the \emph{radius} of $X$ around $x^*$),} \\
    d(x^*, X) &= \min_{y \in X} \HD(x^*, y) \quad\text{(the \emph{distance} from $x^*$ to $X$).}
\end{align*}
Let us formally repeat the definitions of the four problems studied in this paper:

\begin{definition}[Continuous Closest String]
Given a set of $n$ strings $X \subseteq \Sigma^d$, find a string $x^* \in \Sigma^d$ which minimizes the radius $r(x^*, X)$.
\end{definition}

\begin{definition}[Continuous Remotest String]
Given a set of $n$ strings $X \subseteq \Sigma^d$, find a string $x^* \in \Sigma^d$ which maximizes the distance $d(x^*, X)$.
\end{definition}

\begin{definition}[Discrete Closest String]
Given a set of $n$ strings $X \subseteq \Sigma^d$, find a string $x^* \in X$ which minimizes the radius $r(x^*, X)$.
\end{definition}

\begin{definition}[Discrete Remotest String]
Given a set of $n$ strings $X \subseteq \Sigma^d$, find a string $x^* \in X$ which maximizes the distance $d(x^*, X \setminus \set{x^*})$.
\end{definition}

\paragraph{Hardness Assumptions}
In this paper, our lower bounds are conditioned on the following two plausible hypotheses from fine-grained complexity.

\begin{definition}[Strong Exponential Time Hypothesis, SETH~\cite{ImpagliazzoP01,ImpagliazzoPZ01}]
For every $\epsilon > 0$, there is some $k \geq 1$ such that $k$-CNF SAT cannot be solved in time $\Order(2^{(1-\epsilon)n})$.
\end{definition}

\begin{definition}[Hitting Set Conjecture~\cite{AbboudWW16}]
For every $\epsilon > 0$, there is some $c \geq 1$ such that no algorithm can decide in $\Order(n^{2-\epsilon})$ time, whether in two given lists $A, B$ of $n$ subsets of a universe of size $c \log n$, there is a set in the first list that intersects every set in the second list (i.e. a ``hitting set'').
\end{definition}
\section{Continuous Closest String is SETH-Hard} \label{sec:continuous}
In this section we present our fine-grained lower bounds for the continuous Closest and Remotest String problems. We start with a high-level overview of our proof, and then provide the technical details in \cref{sec:continuous:sec:qsat,sec:continuous:sec:balancing,sec:continuous:sec:reduction,sec:continuous:sec:theorem}.

Let us first recall that over \emph{binary} alphabets, the continuous Closest and Remotest String problems are trivially equivalent. The insight is that for any two strings $x, y \in \set{0, 1}^d$ we have that $\HD(x, y) = d - \HD(\overline{x}, y)$ where $\overline{x}$ is the \emph{complement} of $x$ obtained by flipping each bit. From this it easily follows that
\begin{equation*}
    \min_{x^* \in \set{0, 1}^d} \max_{y \in X} \HD(x^*, y) = d - \max_{x^* \in \set{0, 1}^d} \min_{y \in X} \HD(x^*, y).
\end{equation*}
Note that finding a string $x^*$ optimizing the left-hand side is exactly the Closest String problem, whereas finding a string $x^*$ optimizing the right-hand side is exactly the Remotest String problem, and thus both problems are one and the same. For this reason, let us focus our attention for the rest of this section only on the Remotest String problem.

\paragraph{Tight Lower Bound for Remotest String}
Our goal is to establish a lower bound under the Strong Exponential Time Hypothesis. To this end, we reduce a $k$-SAT instance with $N$ variables to an instance of the Remotest String problem with dimension~\makebox{$d = (1 + \order(1)) N$}. In \cref{sec:continuous:sec:qsat,sec:continuous:sec:balancing,sec:continuous:sec:reduction,sec:continuous:sec:theorem} we will actually reduce from a $q$-ary analogue of the $k$-SAT in order to get a tight lower bound for all alphabet sizes $|\Sigma|$. However, for the sake of simplicity we stick to binary strings and the usual $k$-SAT problem in this overview. Our reduction runs in two steps.

\paragraph{Step 1: Massaging the SAT Formula}
In the first step, we bring the given SAT formula into a suitable shape for the reduction to the Remotest String problem. Throughout, we partition the variables $[N]$ into \emph{groups}~$P_1, \dots, P_{\frac Ns}$ of size exactly $s$ (where $s$ is a parameter to be determined later). We assert the following properties:
\begin{itemize}
    \item \emph{Regularity:} All clauses contain exactly $k$ literals, and all clauses contain literals from the same number of groups (say $r$). This property can be easily be guaranteed by adding a few fresh variables to the formula, all of which must be set to $0$ in a satisfying assignment, and by adding these variables to all clauses which do not satisfy the regularity constraint yet.
    \item \emph{Balancedness:} Let us call an assignment $\alpha \in \set{0, 1}^N$ \emph{balanced} if in every group it assigns exactly half the variables to $0$ and half the variables to $1$. We say that a formula is \emph{balanced} if it is either unsatisfiable or if it is satisfiable by a balanced assignment. To make sure that a given formula is balanced, we can for instance flip each variable in the formula with probability \smash{$\frac12$}. In this way we balance each group with probability \smash{$\approx \frac{1}{\sqrt s}$}, and so all $\frac{N}{s}$ groups are balanced with probability at least~\smash{$s^{-\frac{N}{2s}}$}. By choosing~\makebox{$s = \omega(1)$}, this random experiment yields a balanced formula after a negligible number of repetitions. In \cref{lem:qsat-balancing} we present a deterministic implementation of this idea.
\end{itemize}

\paragraph{Step 2: Reduction to Remotest String}
The next step is to reduce a regular and balanced $k$-CNF formula to an instance of the Remotest String problem. The idea is to encode all \emph{falsifying} assignments of the formula as strings---a sufficiently remote point should in spirit be remote from falsifying and thus satisfying. To implement this idea, take any clause $C$ from the instance. Exploiting the natural correspondence between strings and assignments, we add all strings $\alpha \in \set{0, 1}^n$ that satisfy the following two constraints to the Remotest String instance:
\begin{enumerate}
    \item The assignment $\alpha$ falsifies the clause $C$.
    \item For any group $P_i$ that does \emph{not} contain a variable from $C$, we have that $\alpha[P_i] = 0^s$ or $\alpha[P_i] = 1^s$.
\end{enumerate}
We start with the intuition behind the second constraint: For any \emph{balanced} assignment $\alpha$ and any group $P_i$ that does not contain a variable from $C$, we have that $\HD(\alpha^*[P_i], \alpha[P_i]) = \frac s2$ (the string $\alpha^*[P_i]$ contains half zeros and half ones, whereas $\alpha[P_i]$ is either all-zeros or all-ones). There are exactly \smash{$\frac Ns - r$} such groups (by the regularity), leading to Hamming distance \smash{$\frac{s}{2} (\frac{N}{s} - r)$}.

It follows that the only groups that actually matter for the distance between $\alpha^*$ and $\alpha$ are the groups which do contain a variable from $C$. Here comes the first constraint into play: If $\alpha^*$ is a satisfying assignment, then~$\alpha^*$ and $\alpha$ must differ in at least one of these groups and therefore have total distance at least~\smash{$\frac{s}{2} (\frac Ns - r) + 1$}. Conversely, for any falsifying assignment $\alpha^*$ there is some string $\alpha$ in the instance with distance at most \smash{$\frac{s}{2} (\frac Ns - r)$}. Therefore, to decide whether the SAT formula is satisfiable it suffices to compute whether there is a Remotest String with distance at least \smash{$\frac{s}{2} (\frac Ns - r) + 1$}. Finally, it can be checked that the number of strings $\alpha$ added to the instance is manageable.

This completes the outline of our hardness proof, and we continue with the details. In \cref{sec:continuous:sec:qsat} we introduce the $(q, k)$-SAT problem which we will use to give a clean reduction also for alphabet larger than size $2$. In \cref{sec:continuous:sec:balancing} we formally prove how to guarantee that a given $(q, k)$-SAT formula is regular and balanced, and in \cref{sec:continuous:sec:reduction} we give the details about the reduction to the Remotest String problem. We put these pieces together in \cref{sec:continuous:sec:theorem} and formally prove \cref{thm:continuous-remotest-seth}.


\subsection{\boldmath\texorpdfstring{$q$}{q}-ary SAT} \label{sec:continuous:sec:qsat}
To obtain our full hardness result, we base our reduction on the hardness of \emph{$q$-ary} analogue of the classical $k$-SAT problem. We start with an elaborate definition of this problem. Let $X_1, \dots, X_N$ denote some $q$-ary variables (i.e., variables taking values in the domain $[q]$). A \emph{literal} is a Boolean predicate of the form $X_i \neq a$, where $X_i$ is one of the variables and $a \in [q]$. A \emph{clause} is a conjunction of several literals; we say the clause has \emph{width} $k$ if it contains exactly $k$ literals. A \emph{$(q, k)$-CNF formula} is a disjunction of clauses of width at most $k$. Finally, in the $(q, k)$-SAT problem, we are given a $(q, k)$-CNF formula over $M$ clauses and $N$ $q$-ary variables, and the task is to check whether there exists an assignment $\alpha \in [q]^N$ which satisfies all clauses. This problem has already been adressed in previous works~\cite{Traxler08,StephensV19}, and it is known that $q$-ary SAT cannot be solved faster than exhaustive search unless SETH fails:

\begin{lemma}[$q$-ary SAT is SETH-Hard~{{\cite[Theorem~3.3]{StephensV19}}}] \label{lem:qsat-seth}
For any $\epsilon > 0$, there is some~$k \geq 3$ such that for all $q = q(N) \geq 2$, $(q, k)$-SAT cannot be solved in time $\Order(q^{(1-\epsilon) N} \poly(M))$, unless SETH fails.
\end{lemma}

Note that this hardness result applies even when $q$ grows with $N$. We will later exploit this by proving hardness for Remotest String even for alphabets of super-constant size.

\subsection{Regularizing and Balancing} \label{sec:continuous:sec:balancing}
Before we get to the core of our hardness result, we need some preliminary lemmas on the structure of $(q, k)$-CNF formulas. Throughout, let $N$ be the number of variables and let $\mathcal P$ be a partition of $N$ into \emph{groups} of size exactly $s$. (Note that the existence of $P$ implies that $N$ is divisible by $s$.) In two steps we will now formally introduce the definitions of regular and balanced formulas and show how to convert unconstrained formulas into regular and balanced ones.  

\begin{definition}[Regular Formulas]
Let $\phi$ be a $(q, k)$-CNF formula over $N$ variables, and let $\mathcal P$ be a partition of $[N]$. We say that $\phi$ is \emph{$r$-regular} (with respect to $\mathcal P$) if every clause contains exactly $k$ literals from exactly~$r$ distinct groups in $\mathcal P$.
\end{definition}

\begin{lemma}[Regularizing] \label{lem:qsat-regularizing}
Let $\phi$ be a $(q, k)$-CNF formula, and let $2k \leq s \leq N$. In time $\poly(N M)$ we can construct a $(q, 2k)$-CNF formula $\phi'$ satisfying the following properties:
\begin{itemize}
    \item $\phi'$ is satisfiable if and only if $\phi$ is satisfiable.
    \item $\phi'$ has at most $N + \Order(s)$ variables and at most $M + \Order(s \poly(q))$ clauses.
    \item $\phi'$ is $(k+1)$-regular with respect to some partition $\mathcal P$ into groups of size exactly $s$. 
\end{itemize}
\end{lemma}
\begin{proof}
First, observe that there is a $(q, k+1)$-CNF formula $\psi(Y_1, \dots, Y_{k+1})$ on $k+1$ variables with a unique satisfying assignment, say the all-zeros assignment: Include all clauses on $Y_1, \dots, Y_{k+1}$ except for the clause ruling out the all-zeros assignment.

Let us construct a $(q, k+1)$-CNF formula $\phi''$ as follows: Let $X_1, \dots, X_N$ denote the variables of the formula~$\phi$ (after possibly adding at most $s$ unused variables to satisfy the divisibility constraint $s \mid N$), let~\makebox{$Y_{1, 1}, \dots, Y_{s, k+1}$} be $(k+1)s$ fresh variables, and let $\phi'' = \phi(X_1, \dots, X_N) \land \bigwedge_{\ell=1}^s \psi(Y_{i, 1}, \dots, Y_{i, k+1})$. Clearly~$\phi''$ is satisfiable if and only if $\phi$ is satisfiable, it contains $N + (k+1)s = N + \Order(s)$ variables and~$M + s (q^{k+1} - 1) = M + \Order(s \poly(q))$ clauses.

Next, we turn $\phi''$ into a regular formula $\phi'$ by modifying the existing clauses without adding new variables or clauses. Let $\mathcal P$ be a partition of the involved variables that arbitrarily partitions $X_1, \dots, X_N$, and puts each block~$Y_{1, i}, \dots, Y_{s, i}$ into a separate group. In order to make $\phi''$ $(k+1)$-regular, we have to fill up the number of variables in every clause to $2k$, while making sure that each clause has variables from exactly~\makebox{$k+1$} groups. To this end, we add to each clause an appropriate number of literals of the form $Y_{i, j} \neq 0$. It is easy to check that it is always possible to find such literals which increase the width to exactly $2k$ and the number of participating groups to exactly $k + 1$. Moreover, note that these literals cannot be satisfied and therefore do not change the truth value of the formula.
\end{proof}

\begin{definition}[Balanced Formulas]
Let $\mathcal P$ be a partition of $[N]$ into groups of size $s$. We say that an assignment~\makebox{$\alpha \in [q]^N$} is \emph{balanced} (with respect to $\mathcal P$) if in every group of $\mathcal P$, $\alpha$ assigns each symbol in $[q]$ exactly~$\frac{s}{q}$ times. We say that a $(q, k)$-CNF formula $\phi$ is \emph{balanced} (with respect to $\mathcal P$) if either $\phi$ is unsatisfable, or $\phi$ is satisfiable by a balanced assignment~$\alpha$. 
\end{definition}

\begin{lemma}[Balancing] \label{lem:qsat-balancing}
Let $\phi$ be a $(q, k)$-CNF formula over $N$ variables, let $\mathcal P$ be a partition of $[N]$ into groups of size $s$, and assume that $q$ divides $s$. We can construct $(q, k)$-CNF formulas $\phi_1, \dots, \phi_t$ over the same number of variables and clauses as $\phi$ such that:
\begin{itemize}
    \item For all $i \in [t]$, $\phi_i$ is satisfiable if and only if $\phi$ is satisfiable.
    \item There is some $i \in [t]$ such that $\phi_i$ is balanced (with respect to $\mathcal P$).
    \item \smash{$t = ((s + 1)(q-1))^{(q-1)\frac Ns}$}, and we can compute all formulas in time $\poly(N M t)$.
\end{itemize}
\end{lemma}
\begin{proof} 
Let us first focus on a single group $P \in \mathcal P$. We demonstrate how to construct a set $\Phi$ of~$((s + 1)(q-1))^{q-1}$ equivalent formulas with the guarantee that at least one such formula is balanced for this group. These formulas are constructed in stages $\ell = 0, \dots, q - 1$, by inductively proving that there is a set $\Phi_\ell$ satisfying the following three properties (for all $\ell$):
\begin{itemize}
    \item $\Phi_\ell$ consists of formulas that are satisfiable if and only if $\phi$ is satisfiable.
    \item If $\phi$ is satisfiable, then there is some $\phi' \in \Phi_\ell$ and a satisfying assignment $\alpha$ of $\phi'$ that assigns to the variables in $P$ each symbol in $[\ell]$ exactly $\frac sq$ times. That is, $|\alpha^{-1}(1) \cap P| = \dots = |\alpha^{-1}(\ell) \cap P| = \frac sq$. We say that $\phi'$ is \emph{balanced up to $\ell$.}
    \item $|\Phi_\ell| \leq ((s + 1)(q-1))^{\ell}$.
\end{itemize}
We start with $\Phi_0 = \set{\phi}$ for which all of these properties are trivially satisfied. So let $0 < \ell < q$. In order to construct $\Phi_\ell$, we make use of the following insight: We can, without changing the truth value of the formula, arbitrarily permute the $q$ values of any variable as long as we do so consistently across all clauses. To construct the set $\Phi_\ell$, we start from an arbitrary formula $\phi' \in \Phi_{\ell-1}$ and guess a symbol $a \in [q] \setminus [\ell]$. We sweep over all variables $i \in P$ and in each step exchange the values $\ell$ and $a$ of the variable $X_i$. Let~\makebox{$\phi'_{a, 0}, \dots, \phi'_{a, s}$} denote the~$s + 1$ formulas created during this process (with $\phi'_{a, 0} = \phi'$ and $\phi'_{a, s}$ being the formula after exchanging the values $\ell$ and $a$ for all variables in $P$). We assign $\Phi_\ell = \set{\phi'_{a, 0}, \dots, \phi'_{a, s} : \phi' \in \Phi_{\ell-1}, a \in [q] \setminus [\ell]}$.

Let us argue that this choice is correct. The first property is obvious as we never alter the truth value of a formula, and the third property is also simple since $|\Phi_\ell| \leq (s + 1)(q-1) \cdot ((s + 1)(q-1))^{\ell-1}$ by induction. The second property is more interesting. By induction we know that there is some formula $\phi' \in \Phi_{\ell-1}$ that is balanced up to $\ell-1$, for some assignment $\alpha$. By construction the formulas $\phi'_{a, 0}, \dots, \phi'_{a, s}$ are also balanced up to $\ell-1$, and we claim that there is at least one that is balanced up to $\ell$. To see this, let $\alpha_{a, j}$ denote the assignments corresponding to $\alpha$ after performing the same value changes that transform $\phi'$ into $\phi_{a, j}$. Let us further define~\makebox{$k_{a, j} = |\alpha_{a, j}^{-1}(\ell) \cap P|$}. Note that it suffices to prove that there is some pair~$a, j$ with $k_{a, j} = \frac sq$. To find such a pair, distinguish two cases: If \smash{$|\alpha^{-1}(\ell) \cap P| \leq \frac sq$}, then take some~$a$ with \smash{$|\alpha^{-1}(a) \cap P| \geq \frac sq$}, and if \smash{$|\alpha^{-1}(\ell) \cap P| \geq \frac sq$}, then we take some~$a$ with \smash{$|\alpha^{-1}(a) \cap P| \leq \frac sq$}. It is easy to check that such symbols~$a$ exist. Let us focus on the former case; the latter case is symmetric. Since $\phi' = \phi_{a, 0}$ we have that \smash{$k_{a, 0} \leq \frac sq$} and since~$\phi_{a, s}$ is obtained from $\phi'$ by exchanging all occurrences of $\ell$ and $a$ (in the group $P$) we have that~\smash{$k_{a, s} \geq sq$}. Finally, note that any two adjacent values $k_{a, j}$ and $k_{a, j+1}$ differ by at most one. It easily follows that there is indeed some $0 \leq j \leq s$ with \smash{$k_{a, j} = \frac sq$}.

After constructing $\Phi_0, \dots, \Phi_{q-1}$ in this way, let us pick $\Phi = \Phi_{q-1}$. While by induction we only know that this set contains some formula $\phi'$ that is balanced up to $q-1$, by the divisibility constraint that $q$ divides $s$ it follows that the respective assignment also assigns the value $q$ to exactly \smash{$\frac{s}{q}$} variables. This completes the construction of a set $\Phi$ which balances a \emph{single} group.

To complete the lemma, we iterate this construction: Start with a single group $P$ and run the construction to obtain some set $\Phi$. Then, for each formula in $\Phi$ we rerun the construction for a second group and update~$\Phi$ to be the resulting set. After dealing with all $\frac Ns$ groups, we have indeed constructed a balanced formula, and the size of $\Phi$ has reached exactly \smash{$t = ((s + 1)(q - 1))^{(q-1)\frac Ns}$}. As a final detail, we have neglected the running time of this process so far, but it is easy to implement this idea in time $\poly(N M t)$.
\end{proof}

\subsection{Reduction to Remotest String} \label{sec:continuous:sec:reduction}
Having in mind that for our reduction we can assume the SAT formula to be regular and balanced, the following lemma constitutes the core of our reduction:

\begin{lemma}[Reduction from Regular Balanced SAT to Remotest String] \label{lem:qsat-to-remotest}
Suppose there is an algorithm for the continuous Remotest String problem, running in time $\Order(|\Sigma|^{(1-\epsilon)d} \poly(n))$, for some $\epsilon > 0$. Then there is an algorithm that decides whether a given $s$-partitioned $r$-regular $(q, k)$-SAT formula is satisfiable, and runs in time $\Order(q^{(1-\epsilon) N + \Order(s + \frac{N}{s})} \poly(M))$.
\end{lemma}
\begin{proof}
We start with some notation: For a clause $C$, we write $\mathcal P(C) \subseteq \mathcal P$ to address all groups containing a literal from $C$. We start with the construction of the Remotest String instance with alphabet $\Sigma = [q]$ and dimension $d = N$. Here, we make use of the natural correspondence between strings~\makebox{$\alpha \in \Sigma^d$} and assignments~\makebox{$\alpha \in [q]^N$}. In the instance, we add the following strings: For each clause $C$, add all assignments~\makebox{$\alpha \in [q]^N$} to the instance which satisfy the following two constraints:
\begin{enumerate}
    \item The assignment $\alpha$ falsifies the clause $C$.
    \item For each group $P \in \mathcal P \setminus \mathcal P(C)$, the subsequence $\alpha[P]$ contains only one symbol.\\(That is,~$\alpha[P] = a^s$ for some $a \in [q]$.)
\end{enumerate}
We prove that this instance is complete and sound.

\begin{claim}[Completeness] \label{lem:qsat-to-remotest:clm:completeness}
If $\phi$ is satisfiable, then there is some $\alpha^* \in [q]^d$ with $d(\alpha^*, X) > \frac{(q-1)(N - rs)}{q}$.
\end{claim}
\begin{proof}
Since we assume that the formula $\phi$ is satisfiable and balanced, there is a satisfying and balanced assignment $\alpha^*$. To prove that \smash{$d(\alpha^*, X) > \frac{(q-1)(N - rs)}{q}$}, we prove that \smash{$\HD(\alpha^*, \alpha) > \frac{(q-1)(N - rs)}{q}$} for each string~$\alpha$ added to the Remotest String instance. Let $C$ be the clause associated to $\alpha$. From the two conditions on $\alpha$, we get the following two bounds.

By the first condition, $\alpha$ is a falsifying assignment of $C$. In particular, the subsequence \smash{$\alpha[\bigcup_{P \in \mathcal P(C)} P]$} falsifies $C$ (which is guaranteed to contain all variables visible to $C$) falsifies $C$. Since $\alpha^*$ is a satisfying assignment to the whole formula, and in particular to $C$, we must have that \smash{$\alpha^*[\bigcup_{P \in \mathcal P(C)} P] \neq \alpha[\bigcup_{P \in \mathcal P(C)} P]$}, and thus \smash{$\sum_{P \in \mathcal P(C)} \HD(\alpha^*[P], \alpha[P]) \geq 1$}.

By the second condition, for any group $P \in \mathcal P \setminus \mathcal P(C)$, the subsequence $\alpha[P]$ contains only one symbol. Since $\alpha^*$ is balanced, $\alpha^*[P]$ contains that symbol exactly in a $1/q$-fraction of the positions and differs in the remaining ones from \smash{$\alpha[P]$}. It follows that \smash{$\HD(\alpha^*[P], \alpha[P]) = s - \frac{s}{q} = \frac{(q-1) s}{q}$}.

Combining both bounds, we have that
\begin{equation*}
    \HD(\alpha^*, \alpha) = \!\!\!\sum_{P \in \mathcal P(C)}\!\!\! \HD(\alpha^*, \alpha[P]) + \!\!\!\!\!\sum_{P \in \mathcal P \setminus \mathcal P(C)}\!\!\!\!\! \HD(\alpha^*, \alpha[P]) \geq 1 + \parens*{\frac{N}{s} - r} \cdot \frac{(q-1)s}{q} = \frac{(q-1)(N - rs)}{q} + 1,
\end{equation*}
and the claim follows.
\end{proof}

\begin{claim}[Soundness] \label{lem:qsat-to-remotest:clm:soundness}
If $\phi$ is not satisfiable, then for all $\alpha^* \in [q]^d$ we have $d(\alpha^*, X) \leq \frac{(q-1)(N - rs)}{q}$.
\end{claim}
\begin{proof}
Take any $\alpha^* \in [q]^d$. Since $\phi$ is not satisfiable, $\alpha^*$ is a falsifying assignment of $\phi$ and thus there is some clause $C$ that is falsified by $\alpha^*$. Our strategy is to find some string $\alpha \in [q]^d$ in the constructed instance with~\smash{$\HD(\alpha^*, \alpha) \leq \frac{(q-1)(N - rs)}{q}$}.

We define that string $\alpha$ group-wise: In the groups $\mathcal P(C)$ touching~$C$, we define~$\alpha$ to be exactly as~$\alpha^*$, that is, \smash{$\alpha[\bigcup_{P \in \mathcal P(C)}] := \alpha^*[\bigcup_{P \in \mathcal P(C)}]$}. For each group $P \in \mathcal P \setminus \mathcal P(C)$ not touching~$C$, let $a \in [q]$ be an arbitrary symbol occurring at least \smash{$\frac{s}{q}$} times in $\alpha^*[P]$ and assign $\alpha[P] := a^s$. By this construction we immediately have that \smash{$\HD(\alpha[P], \alpha^*[P]) \leq s - \frac{s}{q} = \frac{(q-1)s}{q}$}, and in total
\begin{equation*}
    \HD(\alpha^*, \alpha) = \!\!\!\sum_{P \in \mathcal P(C)}\!\!\! \HD(\alpha^*, \alpha[P]) + \!\!\!\!\!\sum_{P \in \mathcal P \setminus \mathcal P(C)}\!\!\!\!\! \HD(\alpha^*, \alpha[P]) \leq 0 + \parens*{\frac{N}{s} - r} \cdot \frac{(q-1)s}{q} = \frac{(q-1)(N - rs)}{q},
\end{equation*}
as claimed.
\end{proof}

In combination, \cref{lem:qsat-to-remotest:clm:completeness,lem:qsat-to-remotest:clm:soundness} show that the constructed instance of the Remotest String problem is indeed equivalent to the given $(q, k)$-SAT instance $\phi$ in the sense that $\phi$ is satisfiable if and only if there is a remote string with distance more than \smash{$\frac{(q-1)(N - rs)}{q}$}.

\medskip
It remains to analyze the running time. Let $n$ denote the number of strings in the constructed instance. As a first step, we prove that $n \leq q^{\Order(s) + \frac{N}{s}} \cdot M$ and that we can construct the instance in time $\poly(n)$. Indeed, focus on any of the $M$ clauses. The strings $\alpha$ in the instance are unconstrained in all groups touching~$M$ (up to the condition that $\alpha$ must falsify some clause) which accounts for $r \cdot s$ positions and thus $q^{rs} = q^{\Order(s)}$ options. For each group not touching $M$ we can choose between $q$ possible values, and therefore the total number of options is \smash{$q^{\frac{N}{s} - r} \leq q^{\frac{N}{s}}$}. Therefore, the total number of strings is indeed $n \leq M \cdot q^{\Order(s)} \cdot q^{\frac{N}{s}}$. Moreover, it is easy to see check that the instance can be constructed in time $\poly(n)$.

As the time to construct the instance is negligible, the total running time is dominated by solving the Remotest String instance. Assuming an efficient algorithm in time $\Order(|\Sigma|^{(1-\epsilon)d} \poly(n))$, this takes time~\smash{$\Order(q^{(1-\epsilon) N + \Order(s + \frac Ns)} \poly(M))$} as claimed.
\end{proof}

\subsection{Putting the Pieces Together} \label{sec:continuous:sec:theorem}
We are finally ready to prove \cref{cor:continuous-closest-seth,thm:continuous-remotest-seth}.

\corcontinuousclosestseth*
\begin{proof}
From \cref{thm:continuous-remotest-seth} we know that it is SETH-hard to solve the Remotest String problem over a binary alphabet in time $\Order(2^{(1-\epsilon) d} \poly(n))$. For this proof, it thus suffices to reduce a Closest String instance~\makebox{$X \subseteq \set{0, 1}^d$} to an instance of Remotest String without increasing $n$ and $d$. The construction is trivial: Simply view $X$ as an instance of Remotest String. Let us write $\overline x$ to denote the \emph{complement} of a binary string $x$ (obtained by flipping the bits in all positions). Since for any binary strings $x, y$ it holds that
\begin{equation*}
    \HD(x, y) = d - \HD(\overline x, y),
\end{equation*}
it follows that
\begin{equation*}
    \min_{x^* \in \set{0, 1}^d} \max_{y \in X} \HD(x, y) = \min_{x^* \in \set{0, 1}^d} \max_{y \in X} (d - \HD(\overline{x^*}, y)) = d - \max_{x^* \in \set{0, 1}^d} \min_{y \in X} \HD(x^*, y).
\end{equation*}
We can therefore solve the Closest String problem by returning $d$ minus the Remotest String distance.
\end{proof}

\thmcontinuousremotestseth*
\begin{proof}
Suppose that the continuous Remotest String problem can be solved in time $\Order(|\Sigma|^{(1-\epsilon) d} \poly(n))$ for some $\epsilon > 0$ and for $|\Sigma| = \order(d)$. With this in mind, we design a better-than-brute-force $(q, k)$-SAT algorithm for $q = |\Sigma|$ by combining the previous three \cref{lem:qsat-regularizing,lem:qsat-balancing,lem:qsat-to-remotest}. Let $\phi$ be the input formula, and let~$\mathcal P$ denote a partition of the variables into groups of size $s$ (which is yet to be determined) as before.
\begin{enumerate}
    \item Using \cref{lem:qsat-regularizing}, construct a \emph{regular} $(q, 2k)$-formula $\phi'$ which is equivalent to $\phi$.
    \item Using \cref{lem:qsat-balancing}, construct regular $(q, 2k)$-formulas $\phi_1', \dots, \phi_t'$ all of which are equivalent to $\phi$. At least one of these formulas is \emph{balanced}.
    \item By means of the reduction in \cref{lem:qsat-to-remotest}, solve all $t$ formulas $\phi_1', \dots, \phi_t'$. If a formula is reported to be satisfiable, check whether the answer is truthful (e.g., using the standard decision-to-reporting reduction) and if so report that the formula is satisfiable. We need the additional test since, strictly speaking, we have not verified in \cref{lem:qsat-to-remotest} that the algorithm is correct for non-balanced inputs.
\end{enumerate}
The correctness is obvious. Let us analyze the running time. Constructing the formula $\phi'$ takes polynomial time and can be neglected. By \cref{lem:qsat-regularizing}, $\phi'$ has $N' = N + \Order(s)$ variables and $M' = M + \Order(s \poly(q))$ clauses. The construction of the formulas $\phi_1', \dots, \phi_t'$ also runs in polynomial time $\poly(N' M' t)$ and can be neglected; this time we do not increase the number of variables and clauses. Moreover, \cref{lem:qsat-balancing} guarantees that
\begin{equation*}
    t = ((s + 1)(q - 1))^{(q - 1) \frac{N'}{s}} \leq (sq)^{\Order(\frac{qN}{s})},
\end{equation*}
By picking $s = c q$ (for some parameter $c$ to be determined), this becomes
\begin{equation*}
    t \leq (c q^2)^{\Order(\frac{N}{c})} = q^{\Order(\frac{N}{c } \log_q(cq^2))} = q^{N \cdot \Order(\frac{\log c}{c})}.
\end{equation*}
Finally, by \cref{lem:qsat-to-remotest} solving each formula $\phi_i'$ takes time
\begin{equation*}
    q^{(1-\epsilon) N' + \Order(s + \frac{N'}{s})} \poly(M') = q^{(1-\epsilon) N + \Order(s + \frac Ns)} \poly(M) = q^{(1-\epsilon) N + \order(c N) + \Order(\frac Nc)} \poly(M),
\end{equation*}
(using that $s = cg = \order(c N)$), and thus the total running time is bounded by
\begin{equation*}
    q^{N \cdot \Order(\frac{\log c}{c})} \cdot q^{(1-\epsilon) N + \order(c N) + \Order(\frac Nc)} \poly(M) = q^{(1-\epsilon + \order(c) + \Order(\frac{\log c}{c})) N} \poly(M).
\end{equation*}
Note that by picking $c$ to be a sufficiently large constant (depending on $\epsilon$), the exponent becomes $(1 - \frac{\epsilon}{2}) N$, say. We have therefore witness an algorithm for the $(q, k)$-SAT problem in time \smash{$\Order(q^{(1 - \epsilon/2) N} \poly(M))$}, which contradicts SETH by \cref{lem:qsat-seth}.
\end{proof}
\section{Discrete Closest String via Inclusion-Exclusion} \label{sec:discrete}
In this section, we mainly focus on presenting an algorithm for the discrete Closest String problem with \emph{subquadratic} running time whenever the dimension is small, i.e. $d= o(\log n)$. Our algorithm relies on the inclusion-exclusion principle, and is, to the best of our knowledge, the first application of this technique to the Closest and Remotest String problems.  Specifically, we obtain the following result:


\thminclexclclosest*

We structure this section as follows: First, we present a high-level overview of the main ideas behind the algorithm; for the sake of presentation, we focus only on the Closest String problem. We start developing a combinatorial toolkit to tackle the Closest String problem in \cref{sec:discrete:sec:inclusion-exclusion}. Then, in \cref{sec:discrete:sec:algorithm} we provide the actual algorithm and prove \cref{thm:incl-excl-closest}. Finally, in \cref{sec:discrete:sec:remotest} we refer to the discrete Remotest String problem.

After presenting our algorithm for the discrete Closest String problem with \emph{subquadratic} running time whenever the dimension is small in \cref{sec:discrete:sec:large_dim} we present an algorithm for the discrete Closest String problem with better running time in the large-dimension regime. In \cref{sec:discrete:sec:equivalence}, at the end of this section, we show a fine-grained equivalence between the Closest String problem and the Remotest String problem.

Before we describe our algorithm for the discrete Closest String problem, we provide some intuition about the general connection between the inclusion-exclusion principle and the Hamming distance between a pair of strings. Our key insight is that the inclusion-exclusion principle allows us to express whether two strings have Hamming distance bounded by, say $k$. The following lemma makes this idea precise:

\begin{restatable}[Hamming Distance by Inclusion-Exclusion]{lemma}{lemhdinclusionexclusion} \label{lem:hd-inclusion-exclusion}
Let $x$ and $y$ be two strings of length $d$ over some alphabet $\Sigma$, and let $0 \leq k < d$. Then:
\begin{equation*}
    \One(\HD(x, y) \leq k) = \sum_{\substack{I \subseteq [d]\\|I| \geq d - k}} (-1)^{|I| - d + k} \cdot \binom{|I| - 1}{d - k - 1} \cdot \One(x[I] = y[I]).
\end{equation*}
\end{restatable}

Recall that we write $x[I] = y[I]$ to express that the strings $x$ and $y$ are equally restricted to the indices in~$I$. 
The precise inclusion-exclusion-type formula does not matter too much here, but we provide some  intuition for \cref{lem:hd-inclusion-exclusion}  by considering the special cases where $\HD(x,y)=k$ and $\HD(x,y)=k-1$. If $\HD(x, y) = k$, then there is a unique set $I$ of size $d - k$ for which $x[I] = y[I]$. If instead $\HD(x, y) = k - 1$, then there is a unique such set of size $d - k + 1$, and additionally there are $d - k + 1$ such sets of size $d - k$. The scalars \smash{$(-1)^{|I|-d+k} \binom{|I| - 1}{d - k - 1}$} are chosen in such a way that in any case, all these contributions sum up to exactly $1$.

The takeaway from the above lemma is that we can express the proposition that two strings satisfy \makebox{$\HD(x, y) \leq k$} by a linear combination of $2^d$ indicators of the form $\One(x[I] = y[I])$. It is easy to extend this idea further to the following lemma, which is the core of our combinatorial approach:

\begin{restatable}[Radius by Inclusion-Exclusion]{lemma}{lemradiusinclusionexclusion} \label{lem:radius-inclusion-exclusion}
Let $x$ be a string of length $d$ over some alphabet $\Sigma$, let $X$ be a set of strings each of length $d$ over $\Sigma$, and let $0 \leq k < d$. Then $r(x, X) \leq k$ if and only if
\begin{equation*}
    |X| = \sum_{\substack{I \subseteq [d]\\|I| \geq d - k}} (-1)^{|I| - d + k} \cdot \binom{|I| - 1}{d - k - 1} \cdot \abs{\set{y \in X : x[I] = y[I]}}.
\end{equation*}
\end{restatable}

\noindent
Given this lemma, our algorithm for the Closest String problem is easy to state. We proceed in two steps:

\paragraph{Step 1: Partition.}
Precompute, for all $x \in X$ and for all $I \subseteq [d]$, the value $\abs{\set{y \in X : x[I] = y[I]}}$. This can be implemented in time $\Order(n \cdot 2^d \cdot \poly(d))$ by partitioning the input strings $X$ depending on their characters in the range $I$. After computing this partition, we can read the value $\abs{\set{y \in X : x[I] = y[I]}}$ as the number of strings in the same part as $x$. 

\paragraph{Step 2: Inclusion-Exclusion}
We test for each $0 \leq k \leq d$ and $x \in X$, whether $r(x, X) \leq k$ and finally return the best answer. By \cref{lem:radius-inclusion-exclusion} we can equivalently express the event $r(x, X) \leq k$ via
\begin{equation*}
    |X| = \sum_{\substack{I \subseteq [d]\\|I| \geq d - k}} (-1)^{|I| - d + k} \cdot \binom{|I| - 1}{d - k - 1} \cdot \abs{\set{y \in X : x[I] = y[I]}}.
\end{equation*}
 By observing that the sum contains only~$2^d$ terms and noting that we have precomputed the values $\abs{\set{y \in X : x[I] = y[I]}}$, we can evaluate the sum, for a fixed $x$, in time $\Order(2^d \cdot \poly(d))$. In total, across all strings $x \in X$, the running time becomes $\Order(n \cdot 2^d \cdot \poly(d)$). 

Finally, let us briefly comment on the $\poly(d)$ term in the running time. When evaluating the above sum naively, we naturally incur a running time overhead of $\poly(d)$ since the numbers  in the sum need $\Omega(d + \log n)$ bits to be represented. However, this overhead can be circumvented by evaluating the expression in a smarter way. We provide more details in \cref{sec:discrete:sec:algorithm}.\bigskip

\subsection{Inclusion-Exclusion} \label{sec:discrete:sec:inclusion-exclusion}
We start with the proofs of the technical \cref{lem:hd-inclusion-exclusion,lem:radius-inclusion-exclusion}. These proofs rely in turn on two auxiliary lemmas, which we prove first. 

\begin{lemma}
Let $m > 0$ and $\ell > 0$. Then $\sum_{i = 0}^{\ell} {(-1)}^{i} \binom{m + \ell - 1}{m + i - 1} \binom{m + i - 1}{m - 1} = 0$.
\end{lemma}
\begin{proof}
Using that $\binom nk = \frac{n!}{k!\,(n-k)!}$, we can rearrange the equation as follows:
\begin{gather*}\allowdisplaybreaks
    \sum_{i=0}^{\ell} (-1)^i \binom{m + \ell - 1}{m + i - 1} \binom{m + i - 1}{m - 1} 
    = \sum_{i=0}^{\ell} (-1)^i \frac{(m+\ell-1)!}{(m+i-1)!\,(\ell-i)!} \cdot \frac{(m+i-1)!}{(m-1)!\, i!} \\
    \qquad= \sum_{i=0}^{\ell} (-1)^i \frac{(m+\ell-1)!}{\ell!\, (m-1)!} \cdot \frac{\ell!}{i!\,(\ell-i)!} \\
    \qquad=\binom{m+\ell}{\ell} \sum_{i=0}^{\ell} (-1)^i \binom{\ell}{i}.
\end{gather*}
This becomes zero, as can be seen by an application of the binomial theorem $0 = (1 - 1)^\ell = \sum_{i=0}^\ell (-1)^i \binom{\ell}{i}$, provided that $\ell > 0$.
\end{proof}

\begin{lemma} \label{lem:binom}
Let $m > 0$ and $\ell \geq 0$. Then $\sum_{i = 0}^{\ell} {(-1)}^{i} \binom{m + \ell}{m + i} \binom{m + i - 1}{m - 1} = 1$.
\end{lemma}
\begin{proof}
Fix $m > 0$. The proof is by induction on $\ell$. The equation is trivially true for $\ell = 0$ since both sides become $1$. So assume that $\ell > 0$. By an application of Pascal's triangle (\smash{$\binom nk = \binom{n-1}{k-1} + \binom{n-1}{k}$} if $0 < k < n$ and \smash{$\binom nk = \binom{n-1}{k-1}$} if $n = k$), the left-hand side becomes:
\begin{gather*}
    \sum_{i = 0}^{\ell} {(-1)}^{i} \binom{m + \ell}{m + i} \binom{m + i - 1}{m - 1} \\
    \qquad= \sum_{i = 0}^{\ell} {(-1)}^{i} \binom{m + \ell - 1}{m + i - 1} \binom{m + i - 1}{m - 1} + \sum_{i = 0}^{\ell-1} {(-1)}^{i} \binom{m + \ell - 1}{m + i} \binom{m + i - 1}{m - 1}.
\end{gather*}
Using the previous lemma, the first sum is $0$ and by the induction hypothesis, the second sum is $1$. Since the choice of $m$ was arbitrary, the lemma holds for all $m>0$.
\end{proof}

\lemhdinclusionexclusion*
\begin{proof}
We distinguish two cases. First, if $\HD(x,y) > k$ then for every set $I\subseteq[n]$ of size $|I| \geq d - k$ we have that $x[I]\neq y[I]$ (as otherwise the distance would be at most $k$). Thus, the right-hand side of the equation is zero and so is the left-hand side by the assumption that $\HD(x, y) > k$.

Next, assume that $\HD(x, y) \leq k$. Let $\ell$ such that $\HD(x,y) = k-\ell$ and let $J = \set{i : x[i] = y[i]}$ be the set of matching position between $x$ and $y$. Since the Hamming distance between $x$ and $y$ is $k - \ell$, we have that~\makebox{$|J|= d - k + \ell$}. 
Observe that for any $I \subseteq [d]$ with $I \nsubseteq J$ it holds that $x[I] \neq y[I]$. Hence, $\One(x[I]=y[I]) = 1$ if and only if $I \subseteq J$. We can thus rewrite the right-hand side of the equation in the following way:
\begin{gather*}
\sum_{\substack{I \subseteq [d] \\d - k \le |I|}} {(-1)}^{|I| - d + k} \cdot \binom{|I| - 1}{d - k - 1} \cdot \One(x[I]=y[I]) \\
    \qquad=  \sum_{\substack{I \subseteq J \\d - k \le |I|}} {(-1)}^{|I| - d + k} \cdot \binom{|I| - 1}{d - k - 1}  \\ 
    \qquad=\sum_{i = 0}^{\ell} \sum_{\substack{I \subseteq J \\ |I| = d - k + i}} {(-1)}^{i} \binom{d - k + i - 1}{d - k - 1}\\ 
    \qquad= \sum_{i = 0}^{\ell} {(-1)}^{i} \binom{d - k + \ell}{d - k + i} \binom{d - k + i - 1}{d - k - 1}
\intertext{By defining $m = d - k$, we can further rewrite this expression as} 
    \qquad= \sum_{i = 0}^{\ell} {(-1)}^{i} \binom{d - k + \ell}{d - k + i} \binom{d - k + i - 1}{d - k - 1} = \sum_{i = 0}^{\ell} {(-1)}^{i} \binom{m + \ell}{m + i} \binom{m + i - 1}{m - 1}.
\end{gather*}
Finally, from \cref{lem:binom} we get that $\sum_{i = 0}^{\ell} {(-1)}^{i} \binom{m + \ell}{m + i} \binom{m + i - 1}{m - 1}$ is exactly $1$. In summary: If $\HD(x,y) \le k$ then also the right-hand side of the equation in the lemma is equal to $1$. 
\end{proof}

Below we consider an extension of this lemma to \emph{sets} of strings:

\lemradiusinclusionexclusion*
\begin{proof}
We sum over each $x \in X$ for both sides of the equation in  \cref{lem:hd-inclusion-exclusion} and get:
\begin{equation*}
\sum_{y \in X} \One(\HD(x, y) \leq k) = \sum_{y \in X} \sum_{\substack{I \subseteq [d]\\|I| \geq d - k}} (-1)^{|I| - d + k} \cdot \binom{|I| - 1}{d - k - 1} \cdot \One(x[I] = y[I]).
\end{equation*}
The left-hand side becomes $\sum_{y \in X} \One(\HD(x, y) \leq k)$ which is equal to $|X|$ if and only if $r(x, X) \leq k$. The right-hand side can be interpreted as follows:
\begin{gather*}
\sum_{y \in X} \sum_{\substack{I \subseteq [d]\\|I| \geq d - k}} (-1)^{|I| - d + k} \cdot \binom{|I| - 1}{d - k - 1} \cdot \One(x[I] = y[I]) \\ = \sum_{\substack{I \subseteq [d]\\|I| \geq d - k}} (-1)^{|I| - d + k} \cdot \binom{|I| - 1}{d - k - 1} \cdot \abs{\set{y \in X : x[I] = y[I]}}
\end{gather*}
Thus, we get that $r(x, X) \leq k$ if and only if
\begin{equation*}
    |X| = \sum_{\substack{I \subseteq [d]\\|I| \geq d - k}} (-1)^{|I| - d + k} \cdot \binom{|I| - 1}{d - k - 1} \cdot \abs{\set{y \in X : x[I] = y[I]}}. \qedhere
\end{equation*}
\end{proof}

\subsection{The Algorithm for Discrete Closest String in Details} \label{sec:discrete:sec:algorithm}

In this subsection, we provide our algorithms for the discrete Closest String problem. Let us first demonstrate how to precompute $\abs{\set{y \in X : x[I] = y[I]}}$ for all strings $x \in X$ efficiently. 
\begin{lemma}
\label{lem:efficient-set-compution}
We can compute $\abs{\set{y \in X : x[I] = y[I]}}$ for all strings $x \in X$ in time $\Order(n \cdot 2^d)$.
\end{lemma}
\begin{proof}
Our strategy is to compute, for each $I \subseteq [d]$, a partition ${\mathcal{P}}_{I}$ of the set of all strings $X$ such that two strings ${y}_{1}, {y}_{2} \in X$ are in the same part in ${\mathcal{P}}_{I}$ if and only if ${y}_{1}[I] = {y}_{2}[I]$. This is our goal since, for all strings~$x \in X$, the value we are interested in $\abs{\set{y \in X : x[I] = y[I]}}$ is exactly the size of the part $P$ in ${\mathcal{P}}_{I}$ that contains $x$. Thus, if we can efficiently compute, for all $I \subseteq [d]$ and all $x\in X$, the partition ${\mathcal{P}}_{I}$ and the part $P \in {\mathcal{P}}_{I}$ such that $x \in P$ then we have the desired algorithm. 

Computing the partition ${\mathcal{P}}_{I}$ for each subset of $I \subseteq [d]$ when $|I| \leq 1$ is simple: The partition~${\mathcal{P}}_{\emptyset}$ contains just one part which is the entire input set. We also know that \makebox{${\mathcal{P}}_{\{i\}} = \{\{x \in X : x[i] = \sigma\}:\sigma \in \Sigma\}$} for every~\makebox{$0 \le i \le d - 1$}. Thus, we can compute the partitions ${\mathcal{P}}_{\emptyset}$ and ${\mathcal{P}}_{\{i\}}$ for every $0 \le i \le d - 1$ in time~$\Order(n \cdot d)$. The remaining question is how to efficiently compute the partitions ${\mathcal{P}}_{I}$ for each subset of $I \subseteq [d]$ where $|I| \geq 2$. 

The idea is to use dynamic programming in combination with a \emph{partition refinement} data structure. Let us start with some notation: For a partition $\mathcal P$ and a set $S$, we define the \emph{refinement} of $\mathcal P$ by $S$ as the partition $\set{P \cap S, P \setminus S : P \in \mathcal P}$. For two partitions $\mathcal P$ and $\mathcal P'$, we define the refinement of $\mathcal P$ by $\mathcal P'$ by the iterative refinement of all sets $S \in \mathcal P'$. In previous work, Habib, Paul, and Viennot~\cite{HabibPV98} have established a data structure to maintain partitions $\mathcal P$ of some universe $[n]$ that efficiently supports the following two operations:
\begin{itemize}
    \item\emph{Refinement:} We can refine a partition $\mathcal P$ by another partition $\mathcal P'$ in time $\Order(n)$.
    \item\emph{Query:} Given a partition $\mathcal P$ and an element $i \in [d]$, we can find the part $i \in P \in \mathcal P$ in time $\Order(1)$.
\end{itemize}

Given this data structure, our algorithm is simple: We enumerate all sets $I$ in nondecreasing order with respect to their sizes $|I|$. Writing $I = I' \cup \set{i}$ (for some $i \in [d])$, we compute $\mathcal P_I$ as the refinement of the previously computed partitions $\mathcal P_{I'}$ and \smash{$\mathcal P_{\set i}$}. It is straightforward to verify that this algorithm is correct. The running time of each refinement step is $\Order(n)$ and so the total running time is $\Order(n \cdot 2^d)$ as claimed.
\end{proof}

We are finally ready to state our algorithm and prove its correctness using Lemmas \ref{lem:radius-inclusion-exclusion} and \ref{lem:efficient-set-compution}.

\begin{proof}[Proof of \cref{thm:incl-excl-closest}]
First, it is clear that if we test for each $0 \leq k \leq d$ and $x \in X$ whether $r(x, X) \leq k$ then we can find the solution to the discrete Closest String problem. From Lemma \ref{lem:radius-inclusion-exclusion} we know that $r(x, X) \leq k$ if and only if:
\begin{gather*}
|X| = \sum_{\substack{I \subseteq [d]\\|I| \geq d - k}} (-1)^{|I| - d + k} \cdot \binom{|I| - 1}{d - k - 1} \cdot \abs{\set{y \in X : x[I] = y[I]}}.
\end{gather*}
Thus, if we efficiently compute $\abs{\set{y \in X : x[I] = y[I]}}$ for all strings $x \in X$ and efficiently compute the right-hand side of the equation we will have an efficient algorithm for the discrete Closest String problem. We know from Lemma~\ref{lem:efficient-set-compution} that we can precompute $\abs{\set{y \in X : x[I] = y[I]}}$ for all strings $x \in X$ in time~\makebox{$\Order(n \cdot 2^d)$}. Therefore, the only missing part of the algorithm is computing the inclusion-exclusion step in~\makebox{$\Order(n \cdot 2^d)$} time.

\begin{algorithm}[t]
\caption{An algorithm for the discrete Closest String problem in the small-distance regime. See \cref{thm:incl-excl-closest}.}\label{alg:one_center}
\begin{algorithmic}[1]
\State \emph{(Step 1: Precompute $T[x,I] = \abs{\set{y \in X : x[I] = y[I]}}$)}
\State ${\mathcal{P}}_{\emptyset} \gets X$
\State ${\mathcal{P}}_{\{i\}} \gets \{\{x \in X : x[i] = \sigma\}:\sigma \in \Sigma\} \: \: \forall i \in [0,\dots,d - 1]$
\For {$I = {I'} \cup \{i\}$}
    \State ${\mathcal{P}}_{I} \gets \text{refinement of }  {\mathcal{P}}_{I'},{\mathcal{P}}_{\{i\}}$
\EndFor
\For {$x \in X, I \subseteq [d]$}
    \State $T[x,I] \gets |P| \: \text{where} \: x \in P \in {\mathcal{P}}_{I}$
\EndFor

\bigskip
\State \emph{(Step 2: Inclusion-Exclusion)}
\For {$x \in X , I \subseteq [d]$}
    \State $S[x,|I|] \gets S[x,|I|] + T[x,I]$
\EndFor
\For {$k \gets 0,\dots,d - 1$}
    \For{$x \in X$}
        \If{$|X| = \sum_{\ell = d - k}^{d} {(-1)}^{\ell} \cdot \binom{\ell - 1}{d - k - 1} \cdot S[x,\ell]$}
            \State return $x$
        \EndIf
    \EndFor
\EndFor
\State return an arbitrary $x \in X$
\end{algorithmic}
\end{algorithm}

If we naively evaluate the inclusion-exclusion formula the running time becomes $\Omega(n \cdot 2^d \cdot d)$ as the intermediate values need $\Omega(d)$ bits to be represented in memory. However, we observe that inclusion-exclusion formula can indeed be evaluated more efficiently by rewriting it as follows:
\begin{gather*}
    \sum_{\substack{I \subseteq [d]\\|I| \geq d - k}} (-1)^{|I| - d + k} \cdot \binom{|I| - 1}{d - k - 1} \cdot \abs{\set{y \in X : x[I] = y[I]}} \\
    \qquad= \sum_{\ell = d - k}^{d} {(-1)}^{\ell} \cdot \binom{\ell - 1}{d - k - 1} \cdot \sum_{\substack{I \subseteq [d]\\|I| = \ell}} \abs{\set{y \in X : x[I] = y[I]}}.
\end{gather*}
We can precompute \smash{$S[x, \ell] := \sum_{I \subseteq [d], |I| = \ell} \abs{\set{y \in X : x[I] = y[I]}}$} for all strings $x \in X$ and all values~\makebox{$1 \leq \ell \leq d$} before we compute the inclusion exclusion step. Since there are $2^{d}$ different subsets of $[d]$ and since we already have access to the values $\abs{\set{y \in X : x[I] = y[I]}}$, for all strings $x \in X$,  computing $S[x, \ell]$ amounts to time~\makebox{$O(n \cdot 2^d)$}. Afterwards, computing
\begin{equation*}
    \sum_{\ell = d - k}^{d} {(-1)}^{\ell} \cdot \binom{\ell - 1}{d - k - 1} \cdot S[x,\ell]
\end{equation*}
for all strings $x \in X$ and for all $0 \leq k \leq d - 1$ only takes time $O(n \cdot {d}^{3})$. Hence, the total running time of the algorithm is $O(n \cdot 2^d)$.
\end{proof}

We summarize the pseudocode of the algorithm outlined in the proof of \cref{thm:incl-excl-closest} in \cref{alg:one_center}.

\subsection{Remotest String via Inclusion-Exclusion}
\label{sec:discrete:sec:remotest}
So far we only considered inclusion-exclusion algorithm for the Closest String problem. However, all of our algorithm carry over in a straightforward way to the Remotest String problem. The only conceptual difference is that instead of using~\cref{lem:radius-inclusion-exclusion}, we here rely on the following combinatorial lemma:

\begin{lemma}[Remoteness by Inclusion-Exclusion]
\label{lem:remoteness-inclusion-exclusion}
Let $x$ be a string of length $d$ over some alphabet $\Sigma$, let $X$ be a set of strings each of length $d$ over $\Sigma$, and let $0 \leq k < d$. Then $d(x, X \setminus \set{x}) > k$ if and only if
\begin{equation*}
    \One(x \in X) = \sum_{\substack{I \subseteq [d]\\|I| \geq d - k}} (-1)^{|I| - d + k} \cdot \binom{|I| - 1}{d - k - 1} \cdot \abs{\set{y \in X : x[I] = y[I]}}.
\end{equation*}
\end{lemma}

We omit the proof of this lemma as it is essentially the same as the proof of Lemma \ref{lem:radius-inclusion-exclusion}. In the algorithm, the only difference appears in the inclusion-exclusion step: While the loop in the algorithm for the Closest String problem is in ascending order, in the algorithm for the Remotest String problem the loop should go in descending order (as we want to return a string of \emph{maximum} remoteness as opposed to the minimum radius). All in all, we obtain the following theorem:

\begin{theorem}[Discrete Remotest String for Small Dimensions] \label{thm:incl-excl-remotest}
The discrete Remotest String problem can be solved in time $\Order(n \cdot 2^d)$.
\end{theorem}

\subsection{Discrete Closest/Remotest String via Fast Matrix Multiplication}
\label{sec:discrete:sec:large_dim}
In this subsection we are concerned with the Closest and Remotest String problems in the large-dimension regime, i.e., where $d$ is polynomial in $n$. Our goal is to prove the following theorem:
\thmfastmmclosest*

The proof relies on fast matrix multiplication, so let us introduce some notation here. Let $\MM(a, b, c)$ denote the time complexity of multiplying an $a \times b$ by a $b \times c$ matrix. Let $\omega$ be such that $\MM(n, n, n) = \Order(n^{\omega})$ (also called the \emph{exponent of matrix multiplication}); it is known that we can take $\omega < 2.373$~\cite{AlmanW21}. Moreover, by a simple decomposition into square block matrices, we have that $\MM(a, b, c) = \Order(\frac{a b c}{\min\set{a, b, c}^{3-\omega}})$.

\begin{proof}
Let $x_1, \dots, x_n$ denote the strings in the given set $X$. Our goal is to construct the~$n \times n$ matrix $D$ for which $D[i, j] = \HD(x_i, x_j)$. Given this matrix $D$ it is easy to solve the Closest String problem in time~$\Order(n^2)$: Select the row $i$ with smallest maximum entry and report $x_i$ as the solution. 

To compute $D$, let $A$ denote the $n \times d|\Sigma|$ matrix defined by
\begin{equation*}
    A[i, (k, \sigma)] = 
    \begin{cases}
        1 &\text{if $x_i[k] = \sigma$,} \\
        0 &\text{otherwise.}
    \end{cases}
\end{equation*}
That is, $A$ indicates in which positions the symbols $\sigma \in \Sigma$ occur in the input. We claim that $D = d \One - A A^T$ (where $\One$ is the all-ones matrix). Indeed, for any pair $i, j \in [n]$ the Hamming distance between $x_i$ and $x_j$ is exactly $d$ minus the number of positions $k \in [d]$ for which $x_i[k] = x_j[k]$, or equivalently, for which there is some symbol $\sigma \in \Sigma$ with $x_i[k] = \sigma$ and $x_j[k] = \sigma$.

By the previous identity, it suffices to compute $A A^T$. Before we start, let us remark that throughout it is necessary to store $A$ \emph{sparsely}, i.e., in a list of nonzero entries. Otherwise, exhaustively writing down the matrix in time $\Order(n d |\Sigma|)$ would possibly exhaust our time budget (since $|\Sigma|$ can be as large as $nd$); however, the sparse representation only takes size $\Order(nd)$.

To compute $A A^T$, we split $A$ into two submatrices: Let~$A_H$ be the submatrix of $A$ consisting of all columns $(k, \sigma)$ containing at least $n^{1-\epsilon}$ nonzero entries (for some parameter $\epsilon > 0$ to be determined), and let $A_L$ denote the submatrix consisting of the remaining columns. Since we split $A$ column-wise, it holds that~$A A^T = A_H A_H^T + A_L A_L^T$. We evaluate these two matrix products using two different algorithms. For the former product \smash{$A_H A_H^T$} we exploit that the matrix $A_H$ is thin with~\makebox{$\Order(nd / n^{1-\epsilon}) = \Order(n^{\delta+\epsilon})$} columns. For this case we therefore prefer fast matrix multiplication, running in time
\begin{equation*}
    \Order(\MM(n, n^{\delta+\epsilon}, n)) = \Order\parens*{\frac{n^{2+\delta+\epsilon}}{\min\set{n, n^{\delta + \epsilon}}^{3-\omega}}} = \Order(n^{2+\delta+\epsilon- (3 - \omega)\min\set{1, \delta + \epsilon}}).
\end{equation*}
The product \smash{$A_L A_L^T$} on the other hand, we evaluate in time $\Order(n d n^{1-\epsilon})$ by the naive algorithm: For each nonzero entry $A_L[i, (k, \sigma)]$, enumerate the at most $n^{1-\epsilon}$ indices $j$ with~\makebox{$A_L[j, (k, \sigma)] \neq 0$}. The total running time is bounded by
\begin{equation*}
    \Order(n^{2+\delta+\epsilon-(3-\omega)\min\set{1, \delta + \epsilon}} + n^{2+\delta-\epsilon}) \leq \Order(n^{2+\delta+\epsilon-0.6 \min\set{1, \delta}} + n^{2+\delta-\epsilon}),
\end{equation*}
and setting $\epsilon = 0.3 \min\set{1, \delta} > 0$ yields the claimed running time.
\end{proof}

The same algorithms works with only minor modififcations also for the Remotest String problem (the only change is that after precomputing $D$, we return find a row $i$ with largest minimum entry and return the string~$x_i$):

\begin{theorem}[Discrete Remotest String for Large Dimensions] \label{thm:fast-mm-remotest}
For all $\delta > 0$, there is some~$\epsilon > 0$ such that the discrete Remotest String problem with dimension $d = n^\delta$ can be solved in time $\Order(n^{2+\delta-\epsilon})$.
\end{theorem}

\subsection{Equivalence of Closest and Remotest String}
\label{sec:discrete:sec:equivalence}
In this section we provide the missing proof of the fine-grained equivalence between the discrete Closest and Remotest String problems over binary alphabets:

\thmdiscreteequiv*

The proof crucially relies on the following auxiliary lemma:

\begin{lemma}[Explicit Constant-Weight Codes] \label{lem:constant-weight-code}
There is an absolute constant $c > 0$, such that for any~$n \in \mathbb N$ there are strings $c_1, \dots, c_n \in \set{0, 1}^d$ of length $d = c \ceil{\log n}$, satisfying the following two properties:
\begin{itemize}
    \item the Hamming weight of all strings $c_1, \dots, c_n$ is exactly $0.25 d$, and
    \item the Hamming distance between any strings is $\HD(c_i, c_j) \geq 0.37 d$ (for $i \neq j$).
\end{itemize}
We can construct $c_1, \dots, c_n$ in deterministic time $\widetilde\Order(n)$.
\end{lemma}
\begin{proof}
Consider any code \smash{$C:\mathbb{F}_4^{\log n}\to \mathbb{F}_4^{(c/4)\log n}$} whose relative distance is $0.74$ for some universal constant~\makebox{$c\in\mathbb N$}. This can be constructed by the code concatenation of any inner code meeting the Gilbert-Varshamov bound with the Reed-Solomon code as the outer code, and the computation is in $\widetilde\Order(n)$ time, as we can find the inner code by exhaustive search. We think of $C$ as a subset of \smash{$\mathbb{F}_4^{(c/4)\log n}$} of cardinality~$n$, where all pairwise Hamming distances are at least $0.74 (c/4) \log n$. We replace each coordinate of each point in $C$ with its characteristic vector in $\{0,1\}^4$. Therefore we obtain \smash{$\tilde{C}\subseteq \{0,1\}^{c\log n}$} still of cardinality $n$. Note that the Hamming weight of each point in \smash{$\tilde C$} is exactly $(c/4)\log n$, and that all pairs of points in \smash{$\tilde{C}$} are at distance at least $2\cdot 0.74 (c/4) \log n = 0.37 c\log n$.
\end{proof}

\begin{proof}[Proof of \cref{thm:discrete-equiv}]
Suppose there is a fast algorithm for the discrete Remotest String problem; we give a fast algorithm for the Closest String problem. Let $X = \set{x_1, \dots, x_n} \subseteq \set{0, 1}^d$ be the Closest String instance which we will convert to an instance of Remotest String. First, compute the strings $c_1, \dots, c_n$ as in \cref{lem:constant-weight-code}. That is, the strings have dimension $d'' = \Order(\log n)$, the Hamming weight of each string $z_i$ is exactly~$0.25 d''$ and the Hamming distance between any pair is at least $0.37 d''$. Let $r = 10 \ceil{\frac{d}{d''}}$. Then consider the following two types of strings for~\makebox{$i \in [n]$}:
\begin{itemize}
    \item $a_i := x_i \circ c_i^r$ (i.e., $x_i$ followed by $r$ repetitions of $c_i$), and
    \item $b_i := \overline{x_i} \circ 0^{r \cdot d''}$ (i.e., the complement of $x_i$ followed by $r \cdot d''$ zeros).
\end{itemize}
Let $A = \set{a_1, \dots, a_n, b_1, \dots, b_n}$ be the instance of the Remotest String problem with dimension $d' = d + d''$. We claim that $X$ and $A$ are equivalent in the following sense:

\begin{claim} \label{thm:discrete-equiv:clm:closest-to-remotest}
$\displaystyle \max_{a \in A} \min_{b \in A \setminus \set{a}} \HD(a, b) = d + r \cdot 0.25d'' - \min_{x \in X} \max_{y \in X} \HD(x, y)$.
\end{claim}
\begin{proof}
As a first step, we derive the following obvious bounds for the Hamming distances between any pair of strings in $A$:
\begin{gather}
    \HD(a_i, b_j) = \HD(x_i, \overline{x_j}) + r \cdot \HD(c_i, 0^{d''}) = d - \HD(x_i, x_j) + r \cdot 0.25 d'', \label{thm:discrete-equiv:clm:closest-to-remotest:eq:1} \\
    \HD(a_i, a_j) = \HD(x_i, x_j) + r \cdot \HD(c_i, c_j) \geq r \cdot \HD(c_i, c_j) \geq r \cdot 0.37 d'' > d + r \cdot 0.25d''\quad\text{(if $i \neq j$),} \label{thm:discrete-equiv:clm:closest-to-remotest:eq:2} \\
    \HD(b_i, b_j) = \HD(\overline{x_i}, \overline{x_j}) \leq d < r \cdot 0.25d''. \label{thm:discrete-equiv:clm:closest-to-remotest:eq:3}
\end{gather}
The rest of the proof is by the following calculation:
\begin{gather*}
    \max_{a \in A} \min_{b \in A \setminus \set{a}} \HD(a, b) \vphantom{\bigg\rbrace} \\
    \qquad= \max\set*{\max_{i \in [n]} \min_{b \in A \setminus \set{a}} \HD(a_i, b),\, \max_{i \in [n]} \min_{b \in A \setminus \set{a}} \HD(b_i, b)}
\intertext{Using that the distance between any strings $b_i, b_j$ (as computed in \cref{thm:discrete-equiv:clm:closest-to-remotest:eq:3}) is strictly smaller than the distance between any strings $a_i$ and $b \in A \setminus \set{a_i}$ (as computed in \cref{thm:discrete-equiv:clm:closest-to-remotest:eq:1,thm:discrete-equiv:clm:closest-to-remotest:eq:2}), we conclude that this maximum is attained by the left-hand side:}
    \qquad= \max_{i \in [n]} \min_{b \in A \setminus \set{a}} \HD(a_i, b) \vphantom{\bigg\rbrace} \\
    \qquad= \max_{i \in [n]} \min\set*{\min_{j \in [n] \setminus \set{i}} \HD(a_i, a_j),\, \min_{j \in [n]} \HD(a_i, b_j)}
\intertext{Using further that the distance between any strings $a_i$ and $a_j$ for $i \neq j$ (as computed by \cref{thm:discrete-equiv:clm:closest-to-remotest:eq:2}) is strictly larger than the distance between any strings $a_i$ and $b_j$ (as computed by \cref{thm:discrete-equiv:clm:closest-to-remotest:eq:1}), we conclude that the minimum is attained by the right-hand side:}
    \qquad= \max_{i \in [n]} \min_{j \in [n]} \HD(a_i, b_j) \\
    \qquad= \max_{i \in [n]} \min_{j \in [n]} (d - \HD(x_i, x_j) + r \cdot 0.25d'') \\
    \qquad= d + r \cdot 0.25d'' - \min_{i \in [n]} \max_{j \in [n]} \HD(x_i, x_j).
\end{gather*}
This completes the proof of \cref{thm:discrete-equiv:clm:closest-to-remotest}.
\end{proof}

\cref{thm:discrete-equiv:clm:closest-to-remotest} proves that we can indeed solve the given Closest String instance by computing a Remotest String in $A$. Constructing the Remotest String instance takes near-linear time $\widetilde\Order(n d)$, and running the fast algorithm with an input of size $|A| = 2n$ and $d' = d + r d'' = \Order(d + \log n)$ takes time $T(\Order(n), \Order(d + \log n))$. Thus, the total time is as claimed.

\medskip
It remains to prove the converse direction, i.e., to reduce the Remotest String problem to the Closest String problem. Since this reduction is very similar to the previous one, we only give the construction and omit the analogous correctness proof.

Let $A = \set{a_1, \dots, a_n} \subseteq \set{0, 1}^d$ be the given Remotest String instance. As before, we precompute the strings $c_1, \dots, c_n$ of length $d'' = \Order(\log n)$ by \cref{lem:constant-weight-code}. Let $r = 10 \ceil{\frac{d}{d''}}$ be as before, and construct the following two types of strings:
\begin{itemize}
    \item $x_i := a_i \circ 0^{r \cdot d''}$ (i.e., $a_i$ followed by $r \cdot d''$ zeros), and
    \item $y_i := \overline{a_i} \circ c_i^r$ (i.e., the complement of $a_i$ followed by $r$ repetitions of $c_i$).
\end{itemize}
Let $X = \set{x_1, \dots, x_n, y_1, \dots, y_n}$ be the Closest String instance. Then the instances $A$ and $X$ are equivalent in the following sense:

\begin{claim}
$\displaystyle \min_{x \in X} \max_{y \in X} \HD(x, y) = d + r \cdot 0.25d'' - \max_{a \in A} \min_{b \in A \setminus \set{a}} \HD(a, b)$.
\end{claim}

\noindent
One can prove this claim in the same vein as \cref{thm:discrete-equiv:clm:closest-to-remotest}. The running time analysis is also identical, so we omit further details.
\end{proof}

As an immediate corollary from the equivalence of Closest and Remotest String, and the known hardness result under the Hitting Set Conjecture~\cite{AbboudBCCS21} (see \cref{thm:discrete-closest-hsc}), we obtain the following result:

\begin{corollary}[Discrete Remotest String for Super-Logarithmic Dimensions]
The discrete Remotest String problem in dimension $d = \omega(\log n)$ cannot be solved in time $\Order(n^{2-\epsilon})$, for any $\epsilon > 0$, unless the Hitting Set Conjecture fails.
\end{corollary}
\bibliographystyle{plainurl}
\bibliography{refs}

\begin{thebibliography}{10}

\bibitem{AbboudBCCS21}
Amir Abboud, MohammadHossein Bateni, Vincent Cohen{-}Addad, {Karthik {C. S.}},
  and Saeed Seddighin.
\newblock On complexity of 1-center in various metrics.
\newblock {\em CoRR}, abs/2112.03222, 2021.
\newblock URL: \url{https://arxiv.org/abs/2112.03222}.

\bibitem{AbboudWY15}
Amir Abboud, Richard~Ryan Williams, and Huacheng Yu.
\newblock More applications of the polynomial method to algorithm design.
\newblock In Piotr Indyk, editor, {\em Proceedings of the Twenty-Sixth Annual
  {ACM-SIAM} Symposium on Discrete Algorithms, {SODA} 2015, San Diego, CA, USA,
  January 4-6, 2015}, pages 218--230. {SIAM}, 2015.
\newblock \href {https://doi.org/10.1137/1.9781611973730.17}
  {\path{doi:10.1137/1.9781611973730.17}}.

\bibitem{AbboudWW16}
Amir Abboud, Virginia~Vassilevska Williams, and Joshua~R. Wang.
\newblock Approximation and fixed parameter subquadratic algorithms for radius
  and diameter in sparse graphs.
\newblock In Robert Krauthgamer, editor, {\em Proceedings of the Twenty-Seventh
  Annual {ACM-SIAM} Symposium on Discrete Algorithms, {SODA} 2016, Arlington,
  VA, USA, January 10-12, 2016}, pages 377--391. {SIAM}, 2016.
\newblock \href {https://doi.org/10.1137/1.9781611974331.ch28}
  {\path{doi:10.1137/1.9781611974331.ch28}}.

\bibitem{AlmanCW16}
Josh Alman, Timothy~M. Chan, and R.~Ryan Williams.
\newblock Polynomial representations of threshold functions and algorithmic
  applications.
\newblock In Irit Dinur, editor, {\em {IEEE} 57th Annual Symposium on
  Foundations of Computer Science, {FOCS} 2016, 9-11 October 2016, Hyatt
  Regency, New Brunswick, New Jersey, {USA}}, pages 467--476. {IEEE} Computer
  Society, 2016.
\newblock \href {https://doi.org/10.1109/FOCS.2016.57}
  {\path{doi:10.1109/FOCS.2016.57}}.

\bibitem{AlmanW15}
Josh Alman and Ryan Williams.
\newblock Probabilistic polynomials and hamming nearest neighbors.
\newblock In Venkatesan Guruswami, editor, {\em {IEEE} 56th Annual Symposium on
  Foundations of Computer Science, {FOCS} 2015, Berkeley, CA, USA, 17-20
  October, 2015}, pages 136--150. {IEEE} Computer Society, 2015.
\newblock \href {https://doi.org/10.1109/FOCS.2015.18}
  {\path{doi:10.1109/FOCS.2015.18}}.

\bibitem{AlmanW21}
Josh Alman and Virginia~Vassilevska Williams.
\newblock A refined laser method and faster matrix multiplication.
\newblock In D{\'{a}}niel Marx, editor, {\em Proceedings of the 2021 {ACM-SIAM}
  Symposium on Discrete Algorithms, {SODA} 2021, Virtual Conference, January 10
  - 13, 2021}, pages 522--539. {SIAM}, 2021.
\newblock \href {https://doi.org/10.1137/1.9781611976465.32}
  {\path{doi:10.1137/1.9781611976465.32}}.

\bibitem{AlonPY09}
Noga Alon, Rina Panigrahy, and Sergey Yekhanin.
\newblock Deterministic approximation algorithms for the nearest codeword
  problem.
\newblock In Irit Dinur, Klaus Jansen, Joseph Naor, and Jos{\'{e}} D.~P. Rolim,
  editors, {\em Approximation, Randomization, and Combinatorial Optimization.
  Algorithms and Techniques, 12th International Workshop, {APPROX} 2009, and
  13th International Workshop, {RANDOM} 2009, Berkeley, CA, USA, August 21-23,
  2009. Proceedings}, volume 5687 of {\em Lecture Notes in Computer Science},
  pages 339--351. Springer, 2009.
\newblock \href {https://doi.org/10.1007/978-3-642-03685-9\_26}
  {\path{doi:10.1007/978-3-642-03685-9\_26}}.

\bibitem{Chen20}
Lijie Chen.
\newblock On the hardness of approximate and exact (bichromatic) maximum inner
  product.
\newblock {\em Theory Comput.}, 16:1--50, 2020.
\newblock \href {https://doi.org/10.4086/toc.2020.v016a004}
  {\path{doi:10.4086/toc.2020.v016a004}}.

\bibitem{CohenHLL97}
Gérard Cohen, Iiro Honkala, Simon Litsyn, and Antoine Lobstein.
\newblock {\em Covering Codes}.
\newblock ISSN. Elsevier Science, 1997.

\bibitem{MenesesLOP04}
Cl{\'{a}}udio~Nogueira de~Meneses, Zhaosong Lu, Carlos A.~S. Oliveira, and
  Panos~M. Pardalos.
\newblock Optimal solutions for the closest-string problem via integer
  programming.
\newblock {\em {INFORMS} J. Comput.}, 16(4):419--429, 2004.
\newblock \href {https://doi.org/10.1287/ijoc.1040.0090}
  {\path{doi:10.1287/ijoc.1040.0090}}.

\bibitem{DopazoRSS93}
Joaqu{\'{\i}}n Dopazo, A.~Rodriguez, J.~C. Saiz, and Francisco Sobrino.
\newblock Design of primers for {PCR} amplification of highly variable genomes.
\newblock {\em Comput. Appl. Biosci.}, 9(2):123--125, 1993.
\newblock \href {https://doi.org/10.1093/bioinformatics/9.2.123}
  {\path{doi:10.1093/bioinformatics/9.2.123}}.

\bibitem{FrancesL97}
Moti Frances and Ami Litman.
\newblock On covering problems of codes.
\newblock {\em Theory Comput. Syst.}, 30(2):113--119, 1997.
\newblock \href {https://doi.org/10.1007/s002240000044}
  {\path{doi:10.1007/s002240000044}}.

\bibitem{GasieniecJL99}
Leszek Gasieniec, Jesper Jansson, and Andrzej Lingas.
\newblock Efficient approximation algorithms for the hamming center problem.
\newblock In Robert~Endre Tarjan and Tandy~J. Warnow, editors, {\em Proceedings
  of the Tenth Annual {ACM-SIAM} Symposium on Discrete Algorithms, 17-19
  January 1999, Baltimore, Maryland, {USA}}, pages 905--906. {ACM/SIAM}, 1999.
\newblock URL: \url{http://dl.acm.org/citation.cfm?id=314500.315081}.

\bibitem{GrammHN02}
Jens Gramm, Falk Huner, and Rolf Niedermeier.
\newblock Closest strings, primer design, and motif search.
\newblock In {\em Sixth Annual International Conference on Computational
  Molecular Biology}, 06 2002.

\bibitem{GrammNR03}
Jens Gramm, Rolf Niedermeier, and Peter Rossmanith.
\newblock Fixed-parameter algorithms for {CLOSEST} {STRING} and related
  problems.
\newblock {\em Algorithmica}, 37(1):25--42, 2003.
\newblock \href {https://doi.org/10.1007/s00453-003-1028-3}
  {\path{doi:10.1007/s00453-003-1028-3}}.

\bibitem{GuruswamiMR04}
Venkatesan Guruswami, Daniele Micciancio, and Oded Regev.
\newblock The complexity of the covering radius problem on lattices and codes.
\newblock In {\em 19th Annual {IEEE} Conference on Computational Complexity
  {(CCC} 2004), 21-24 June 2004, Amherst, MA, {USA}}, pages 161--173. {IEEE}
  Computer Society, 2004.
\newblock \href {https://doi.org/10.1109/CCC.2004.1313831}
  {\path{doi:10.1109/CCC.2004.1313831}}.

\bibitem{HabibPV98}
Michel Habib, Christophe Paul, and Laurent Viennot.
\newblock A synthesis on partition refinement: {A} useful routine for strings,
  graphs, boolean matrices and automata.
\newblock In Michel Morvan, Christoph Meinel, and Daniel Krob, editors, {\em
  {STACS} 98, 15th Annual Symposium on Theoretical Aspects of Computer Science,
  Paris, France, February 25-27, 1998, Proceedings}, volume 1373 of {\em
  Lecture Notes in Computer Science}, pages 25--38. Springer, 1998.
\newblock \href {https://doi.org/10.1007/BFb0028546}
  {\path{doi:10.1007/BFb0028546}}.

\bibitem{HavivR12}
Ishay Haviv and Oded Regev.
\newblock Hardness of the covering radius problem on lattices.
\newblock {\em Chic. J. Theor. Comput. Sci.}, 2012, 2012.
\newblock URL:
  \url{http://cjtcs.cs.uchicago.edu/articles/2012/4/contents.html}.

\bibitem{ImpagliazzoP01}
Russell Impagliazzo and Ramamohan Paturi.
\newblock On the complexity of k-sat.
\newblock {\em J. Comput. Syst. Sci.}, 62(2):367--375, 2001.
\newblock \href {https://doi.org/10.1006/jcss.2000.1727}
  {\path{doi:10.1006/jcss.2000.1727}}.

\bibitem{ImpagliazzoPZ01}
Russell Impagliazzo, Ramamohan Paturi, and Francis Zane.
\newblock Which problems have strongly exponential complexity?
\newblock {\em J. Comput. Syst. Sci.}, 63(4):512--530, 2001.
\newblock \href {https://doi.org/10.1006/jcss.2001.1774}
  {\path{doi:10.1006/jcss.2001.1774}}.

\bibitem{KochmanMP12}
Yuval Kochman, Arya Mazumdar, and Yury Polyanskiy.
\newblock The adversarial joint source-channel problem.
\newblock In {\em Proceedings of the 2012 {IEEE} International Symposium on
  Information Theory, {ISIT} 2012, Cambridge, MA, USA, July 1-6, 2012}, pages
  2112--2116. {IEEE}, 2012.
\newblock \href {https://doi.org/10.1109/ISIT.2012.6283735}
  {\path{doi:10.1109/ISIT.2012.6283735}}.

\bibitem{Lanctot04}
J.~Kevin Lanct{\^{o}}t.
\newblock {\em Some String Problems in Computational Biology}.
\newblock PhD thesis, University of Waterloo, 2004.

\bibitem{LanctotLMWZ03}
J.~Kevin Lanct{\^{o}}t, Ming Li, Bin Ma, Shaojiu Wang, and Louxin Zhang.
\newblock Distinguishing string selection problems.
\newblock {\em Inf. Comput.}, 185(1):41--55, 2003.
\newblock \href {https://doi.org/10.1016/S0890-5401(03)00057-9}
  {\path{doi:10.1016/S0890-5401(03)00057-9}}.

\bibitem{LiMW02}
Ming Li, Bin Ma, and Lusheng Wang.
\newblock On the closest string and substring problems.
\newblock {\em J. {ACM}}, 49(2):157--171, 2002.
\newblock \href {https://doi.org/10.1145/506147.506150}
  {\path{doi:10.1145/506147.506150}}.

\bibitem{LucasBMT91}
K.~Lucas, M.~Busch, S.~Mossinger, and J.~A. Thompson.
\newblock An improved microcomputer program for finding gene- or gene
  family-specific oligonucleotides suitable as primers for polymerase chain
  reactions or as probes.
\newblock {\em Comput. Appl. Biosci.}, 7(4):525--529, 1991.
\newblock \href {https://doi.org/10.1093/bioinformatics/7.4.525}
  {\path{doi:10.1093/bioinformatics/7.4.525}}.

\bibitem{MaS09}
Bin Ma and Xiaoming Sun.
\newblock More efficient algorithms for closest string and substring problems.
\newblock {\em {SIAM} J. Comput.}, 39(4):1432--1443, 2009.
\newblock \href {https://doi.org/10.1137/080739069}
  {\path{doi:10.1137/080739069}}.

\bibitem{MauchMH03}
Holger Mauch, Michael~J. Melzer, and John~S. Hu.
\newblock Genetic algorithm approach for the closest string problem.
\newblock In {\em 2nd {IEEE} Computer Society Bioinformatics Conference, {CSB}
  2003, Stanford, CA, USA, August 11-14, 2003}, pages 560--561. {IEEE} Computer
  Society, 2003.
\newblock \href {https://doi.org/10.1109/CSB.2003.1227407}
  {\path{doi:10.1109/CSB.2003.1227407}}.

\bibitem{MazumdarPS13}
Arya Mazumdar, Yury Polyanskiy, and Barna Saha.
\newblock On chebyshev radius of a set in hamming space and the closest string
  problem.
\newblock In {\em Proceedings of the 2013 {IEEE} International Symposium on
  Information Theory, Istanbul, Turkey, July 7-12, 2013}, pages 1401--1405.
  {IEEE}, 2013.
\newblock \href {https://doi.org/10.1109/ISIT.2013.6620457}
  {\path{doi:10.1109/ISIT.2013.6620457}}.

\bibitem{Micciancio04}
Daniele Micciancio.
\newblock Almost perfect lattices, the covering radius problem, and
  applications to ajtai's connection factor.
\newblock {\em {SIAM} J. Comput.}, 34(1):118--169, 2004.
\newblock \href {https://doi.org/10.1137/S0097539703433511}
  {\path{doi:10.1137/S0097539703433511}}.

\bibitem{ProutskiH96}
V.~Proutski and Edward~C. Holmes.
\newblock Primer master: a new program for the design and analysis of {PCR}
  primers.
\newblock {\em Comput. Appl. Biosci.}, 12(3):253--255, 1996.
\newblock \href {https://doi.org/10.1093/bioinformatics/12.3.253}
  {\path{doi:10.1093/bioinformatics/12.3.253}}.

\bibitem{StephensV19}
Noah Stephens{-}Davidowitz and Vinod Vaikuntanathan.
\newblock Seth-hardness of coding problems.
\newblock In David Zuckerman, editor, {\em 60th {IEEE} Annual Symposium on
  Foundations of Computer Science, {FOCS} 2019, Baltimore, Maryland, USA,
  November 9-12, 2019}, pages 287--301. {IEEE} Computer Society, 2019.
\newblock \href {https://doi.org/10.1109/FOCS.2019.00027}
  {\path{doi:10.1109/FOCS.2019.00027}}.

\bibitem{Traxler08}
Patrick Traxler.
\newblock The time complexity of constraint satisfaction.
\newblock In Martin Grohe and Rolf Niedermeier, editors, {\em Parameterized and
  Exact Computation, Third International Workshop, {IWPEC} 2008, Victoria,
  Canada, May 14-16, 2008. Proceedings}, volume 5018 of {\em Lecture Notes in
  Computer Science}, pages 190--201. Springer, 2008.
\newblock \href {https://doi.org/10.1007/978-3-540-79723-4\_18}
  {\path{doi:10.1007/978-3-540-79723-4\_18}}.

\bibitem{WangCLC06}
Ying Wang, Wei Chen, Xu~Li, and Bing Cheng.
\newblock Degenerated primer design to amplify the heavy chain variable region
  from immunoglobulin cdna.
\newblock {\em {BMC} Bioinform.}, 7({S-4}), 2006.
\newblock \href {https://doi.org/10.1186/1471-2105-7-S4-S9}
  {\path{doi:10.1186/1471-2105-7-S4-S9}}.

\bibitem{Williams18}
Ryan Williams.
\newblock On the difference between closest, furthest, and orthogonal pairs:
  Nearly-linear vs barely-subquadratic complexity.
\newblock In Artur Czumaj, editor, {\em Proceedings of the Twenty-Ninth Annual
  {ACM-SIAM} Symposium on Discrete Algorithms, {SODA} 2018, New Orleans, LA,
  USA, January 7-10, 2018}, pages 1207--1215. {SIAM}, 2018.
\newblock \href {https://doi.org/10.1137/1.9781611975031.78}
  {\path{doi:10.1137/1.9781611975031.78}}.

\bibitem{YusterZ05}
Raphael Yuster and Uri Zwick.
\newblock Fast sparse matrix multiplication.
\newblock {\em {ACM} Trans. Algorithms}, 1(1):2--13, 2005.
\newblock \href {https://doi.org/10.1145/1077464.1077466}
  {\path{doi:10.1145/1077464.1077466}}.

\end{thebibliography}

\end{document}